\documentclass[conference]{IEEEtran}
\ifCLASSINFOpdf
\else
\fi

\usepackage{algorithm}
\usepackage{float}
\usepackage{footmisc}
\usepackage{subfigure} 
\usepackage{array}
\usepackage{enumitem}
\usepackage{verbatim}
\usepackage[table,xcdraw]{xcolor}
\usepackage{tikz}
\usetikzlibrary{calc}
\usetikzlibrary{shapes,arrows}
\tikzstyle{decision} = [diamond, draw, fill=blue!20,
text width=4.5em, text badly centered, node distance=3cm, inner sep=0pt]
\tikzstyle{block} = [rectangle, draw, fill=blue!20,
text width=5em, text centered, rounded corners, minimum height=4em]
\tikzstyle{line} = [draw, -latex']
\tikzstyle{cloud} = [draw, ellipse,fill=red!20, node distance=3cm,
minimum height=2em]
\definecolor{initstepcolor}{RGB}{204, 255, 204}
\definecolor{dynamicinteractionscolor}{RGB}{255, 229, 204}
\definecolor{decisioncolor}{RGB}{255, 179, 179}
\definecolor{boxcolor}{RGB}{198, 235, 251}

\usepackage{varwidth}
\usepackage[noend]{algpseudocode}
\usepackage[normalem]{ulem}
\usepackage[normalem]{ulem}
\usepackage{multirow}
\usepackage{mathtools}
\usepackage{amsthm}
\usepackage{soul}

\newtheorem{mydef}{Definition}


\newtheorem{theorem}{Theorem}[section]
\newtheorem{lemma}[theorem]{Lemma}

\newcommand{\squishlist}{
	\begin{list}{$\bullet$}
		{ \setlength{\itemsep}{0pt}
			\setlength{\parsep}{3pt}
			\setlength{\topsep}{3pt}
			\setlength{\partopsep}{0pt}
			\setlength{\leftmargin}{1.5em}
			\setlength{\labelwidth}{1em}
			\setlength{\labelsep}{0.5em} } }
	
	\newcommand{\squishlisttwo}{
		\begin{list}{$\bullet$}
			{ \setlength{\itemsep}{0pt}
				\setlength{\parsep}{0pt}
				\setlength{\topsep}{0pt}
				\setlength{\partopsep}{0pt}
				\setlength{\leftmargin}{2em}
				\setlength{\labelwidth}{1.5em}
				\setlength{\labelsep}{0.5em} } }
		
		\newcommand{\squishend}{
		\end{list}  }

		\newfont{\mycrnotice}{ptmr8t at 7pt}
		\newfont{\myconfname}{ptmri8t at 7pt}
\begin{document}
%
\title{MRAttractor: Detecting Communities from Large-Scale Graphs}

\author{\IEEEauthorblockN{Nguyen Vo, Kyumin Lee, Thanh Tran}
	\IEEEauthorblockA{Computer Science Department, Worcester Polytechnic Institute \\ Worcester, MA, USA\\
		Email: \{nkvo,kmlee,tdtran\}@wpi.edu}
	
}


%


\maketitle

\begin{abstract}
	Detecting groups of users, who have similar opinions, interests, or social behavior, has become an important task for many applications. A recent study showed that dynamic distance based Attractor, a community detection algorithm, outperformed other community detection algorithms such as Spectral clustering, Louvain and Infomap, achieving higher Normalized Mutual Information (NMI) and Adjusted Rand Index (ARI). However, Attractor often takes long time to detect communities, requiring many iterations. To overcome the drawback and handle large-scale graphs, in this paper we propose MRAttractor, an advanced version of Attractor to be runnable on a MapReduce framework. In particular, we (i) apply a sliding window technique to reduce the running time, keeping the same community detection quality; (ii) design and implement the Attractor algorithm for a MapReduce framework; and (iii) evaluate MRAttractor's performance on synthetic and real-world datasets. Experimental results show that our algorithm significantly reduced running time and was able to handle large-scale graphs.
\end{abstract}
	

%
\IEEEpeerreviewmaketitle

\section{Introduction}
\label{sec:intro}
A community can be viewed as a group of vertices which are densely connected, when compared to the rest of a network. Detecting organizational groups of these vertices paves the way for understanding the underlying structure of complex networks. As a result, there are numerous algorithms proposed in the last decade such as \cite{newman2004detecting, von2007tutorial, blondel2008fast,rosvall2008maps}. Recently, Shao et al. \cite{kdd.dyna} proposed an algorithm called Attractor which utilized the viewpoint of dynamic distance of linked nodes in a graph to find high-quality communities. Instead of optimizing a specific objective, Attractor relied on three types of interactions to dynamically change the distance between vertices. Despite of its superiority over multiple baselines such as Louvain, Ncut and Infomap, this algorithm takes many iterations to converge edge distances, or it even may not converge in some cases \cite{meng2016improved}.

Furthermore, there has been growing interest to process large-scale graphs such as online social networks (e.g., Facebook friendship network and Twitter follower network), and to extract communities. However, most existing community detection algorithms such as Attractor was not be able to handle the large graphs or was designed for a single machine.

In this paper, we were interested in improving Attractor, so that it can handle large-scale graphs,  producing high quality communities as quick as possible with ensuring edge distances converged. Initially, we considered to design, improve, and implement Attractor on well-known graph processing frameworks such as Graphx, Pregel and Pegasus \cite{gonzalez2014graphx,malewicz2010pregel,kyrola2012graphchi,kang2009pegasus,gonzalez2012powergraph}. But these graph processing frameworks are not well suited for communication between unconnected nodes \cite{gonzalez2014graphx} because they share a similar vertex-centric paradigm where only connected vertices can directly communicate with each other. Therefore, we decided to design and implement Attractor on top of a well-known MapReduce framework, Hadoop system \cite{white2012hadoop} which can take advantage of distributed computing power.

However, we faced the following key challenges when we began designing our \textsc{MRAttractor}, an advanced version of Attractor for Hadoop: (i) how to compute dynamic interactions in a distributed computing environment, where a partial graph was loaded to each slave node; (ii) how to force edge distances to converge with minimum overhead of network communication and disk I/O, when edge distances are fluctuated over time or convergence takes long time in some datasets \cite{meng2016improved}; and (iii) how to mitigate the skewness issue of parallel computing (i.e., a task in a slave node takes longer running time than tasks in the other slave nodes). Especially, researchers observed that the original Attractor takes long time in some datasets due to fluctuated edge distances \cite{meng2016improved}. If we just implement Attractor for Hadoop, we will face the same problem which will cause large overhead in the Hadoop system.

By overcoming and resolving these challenges, in this paper, we propose \textsc{MRAttractor} which consists of three main components: (i) Jaccard distance initialization; (ii) dynamic interaction computation by proposing our graph partition algorithms and applying a sliding window technique; and (iii) community extraction.

Our contributions in this paper are as follows:
\squishlist
\item We applied a sliding window technique to ensure edges converged, reduce running time, and still achieve the same quality of extracted communities compared with Attractor. 
\item We designed and implemented \textsc{MRAttractor}, an improved version of Attractor, to reduce running time and handle large-scale graphs. 
\item We evaluated performance of \textsc{MRAttractor} in both synthetic and real-life datasets. Our results showed that \textsc{MRAttractor} was able to handle large-scale graphs, and significantly reduced running time. 
\item We publicly shared our source code and datasets available at http://bit.ly/mrattractor for the research community.
\squishend
\section{Related work}
\label{sec:related_work}

Community detection has been studied for a long time to unveil hierarchical structure and hidden modules of complex networks. To detect communities, many different algorithms were proposed \cite{girvan2002community, fortunato2016community} categorized into (i) statistical inference based methods \cite{karrer2011stochastic}, (ii) optimization-based methods in which they are often designed to optimize a specific objective such as modularity \cite{blondel2008fast}, normalized cut \cite{von2007tutorial} and betweenness \cite{girvan2002community}, and (iii) dynamical processes based algorithms \cite{rosvall2008maps}. To complement the existing approaches, \cite{kdd.dyna} recently proposed an algorithm called \emph{Attractor}, which is based on dynamic distance between linked nodes. This algorithm has been investigated and extended in \cite{meng2016improved, vorevealing}. Despite \emph{Attractor}'s high precision, it was less efficient, requiring many iterations to converge \cite{meng2016improved}.


Other researchers focused on detecting communities from large-scale graphs. One direction was to design and implement algorithms for a MapReduce framework. Tsironis et al. \cite{tsironis2013accurate} proposed a MapReduce spectral clustering algorithm by employing eigensolver and parallel k-means algorithm \cite{zhao2009parallel}. Louvain \cite{louvainMR}, and community detection algorithms based on Label propagation \cite{lpmr} or propinquity dynamics \cite{zhang2009parallel} were developed for Hadoop. Another direction was to design and implement community detection algorithms for other frameworks. For example, \cite{saltz2015distributed, ling2016fast} developed community detection algorithms based on vertex-centric paradigm of Pregel \cite{malewicz2010pregel}. In particular, Saltz et al. \cite{saltz2015distributed} developed an algorithm to optimize Weighted Community Clustering metric. Ling et al. \cite{ling2016fast} proposed modularity-based algorithm called FastCD on top of GraphX. Another work \cite{wickramaarachchi2014fast} employed PMETIS to parallelize the first iteration of Louvain algorithm.



\section{Background: Attractor}
\label{sec:attractor}
In this section, we briefly summarize how Attractor \cite{kdd.dyna} works as the background knowledge so that readers can follow how our MRAttractor works in the following section. Table~\ref{tbl:tableNotation} presents frequently used notations in the rest of this paper.

\begin{table}[t]
	\centering
	\resizebox{0.5\textwidth}{!}
		{
	\begin{tabular}{|l|p{6cm}|}
		\hline
		Notations                                                                     & Meaning                                                                                                                                             \\ \hline
		$G=(V,E)$                                                                     & The undirected graph inputted to \textsc{MRAttractor}                                                                                                           \\
		$n,m$                                                                         & $n=|V|$ and $m=|E|$ of the input graph                                                                                                              \\
		$d(u,v) $                                                                     & Distance of edge $(u,v)$                                                                                                                            \\
		$\Phi(u) $                                                                    & $u$'s neighbors, \small$\Phi(u) = \{v|v \in V, (u,v) \in E \}$                                                                                      \\
		$\Gamma(u)$                                                                   & \begin{tabular}[c]{@{}l@{}}$u$'s neighbors and associated distances\\ $\Gamma(u)=\{(v,{d}(u,v))|u,v \in V,(u,v) \in E \}$\end{tabular}              \\
		$(u, \Gamma(u))$                                                              & The star graph with center vertex $u$.                                                                                                              \\
		$deg(u)$                                                                      & degree of vertex $u$, $deg(u)=|\Phi(u)|=|\Gamma(u)|$                                                                                                \\
		$\big \langle k;v \big \rangle$                                               & A key-value pair where $k$ is key and $v$ is value.                                                                                                 \\
		\begin{tabular}[c]{@{}l@{}}$DI(u,v)$, $CI(u,v)$ \\ and $EI(u,v)$\end{tabular} & \begin{tabular}[c]{@{}l@{}}Direct, common and exclusive interaction \\ between linked nodes $u$ and $v$ respectively\end{tabular}       \\
		$P(\cdot)$	 & Hash function for graph partition\\														
		$\bigtriangleup(u,v,c)$                                                       & Triangle, three edges $(u,v), (u,c), (v,c) \in E$.                                                                                             \\
		$\wedge(u,v,x)$    & Wedge, where $(u,v), (v,x) \in E$, $(u,x)\notin E$.                                                                                                        \\
		$p$                                                                           & The number of partitions of graph $G(V,E)$.                                                                                                         \\
		$\textbf{w}$                                                                  & the sliding window vector of an edge $(u,v)$                                                                                                        \\
		$\lambda$, $\tau$ and $\gamma$    & Cohesive parameter, threshold of sliding window, upper-bound of non-converged edges respectively. \\ \hline
	\end{tabular}
}
		\vspace{2pt}
		\caption{The notations used in this paper}
		\label{tbl:tableNotation}
		\vspace{-25pt}
\end{table}

Attractor consists of three main steps. Firstly, it initializes Jaccard distance of directly linked nodes as follows:
\begin{equation}
d(u,v)=1 - {|N(u) \cap N(v)|}/{|N(u)  \cup  N(v)|}
\label{eq:init-jaccardsim}
\end{equation}
, where $N(u)=\Phi(u)+\{u\}$ and $N(v)=\Phi(v) + \{v\}$. $\Phi(u)$ and $\Phi(v)$ are $u$'s neighbors and $v$'s neighbors respectively.

Secondly, it dynamically changes edges' distance by computing direct linked interaction (DI), common interaction (CI) and exclusive interaction (EI). These interactions are called Dynamic Interactions. The idea behind dynamic interactions is that the more a pair of vertices interacts with each other, the more their distance is reduced (i.e., they attract each other).

$DI(u,v)$ measures the direct influence of linked nodes and is defined based on $sin()$, the sine function, as follows:
\begin{equation}
\small
DI(u,v) = {sin(1-d(u,v)) /  deg(u)} + {sin(1-d(u,v)) /  deg(v)} \label{eq:DI-1}
\end{equation}

$CI(u,v)$ measures influence from common neighbors $c$ of $u$ and $v$, denoted as $CN(u,v) = \Phi(u) \cap \Phi(v)$. Its main concept is if each $c\in CN(u,v)$ has a small $d(c,u)$ and small $d(c,v)$, \emph{u} and \emph{v} will be likely to be in a group.
\begin{equation}
\small
CI(u,v)=  \sum_{\mathclap{c \in CN(u,v)}} CI_c(u,v)
\end{equation}
where $CI_c(u,v)$ is equal to following expression:
\begin{equation}
\label{eq:CI}
\resizebox{1.05\hsize}{!}{
	$ {{(1-d(v,c)) \cdot sin(1-d(u,c))} \over deg(u)} +
	{{(1-d(u,c)) \cdot sin(1-d(v,c)) } \over deg(v)} $
}
\end{equation}


$EI(u,v)$ measures influence from exclusive neighbors. Its main concept is that each exclusive neighbor $x$ of $v$ attracts $v$ to move toward $x$. If $x$ and $u$ has high similarity, the movement of $v$ to $x$ will reduce $d(u,v)$. Otherwise, the distance will increase. The same concept applies to each exclusive neighbor $y$ of $u$. EI of $u$ and $v$ is measured as follows:
\begin{equation}
EI(u,v)=\sum_{\mathclap{x\in EN(v)}} { EI_x(u,v) } +\sum_{\mathclap{y\in EN(u)}} { EI_y(u,v)} \label{eq:EI}
\end{equation}
, where EN($v$) and EN($u$) are sets of exclusive neighbors of $v$ and $u$ respectively. $EN(v)=\Phi(v) -
(\Phi(u)\cap\Phi(v))$ and $EN(u)=\Phi(u) - (\Phi(u)\cap\Phi(v))$. $EI_x(u,v)$ is defined below:
\begin{equation}\label{eq:EI-partial}
EI_{x}(u,v)={ {\rho(x,u) \cdot sin(1-{d}(v,x))} / deg(v)}
\end{equation}
, where $\rho(x,u)$ is influence of vertex $x$ on $d(u,v)$. Given cohesive parameter $\lambda \in [0;1]$, $\rho(x,u)$ is computed based on $\vartheta(x,u)$, the similarity of unconnected nodes $x$ and $u$:
\begin{equation}\label{eq:rho}
\rho(x, u) =
\begin{cases}
\vartheta(x, u), 			& 		\text{if}\ \vartheta(x, u) \geq \lambda \\
\vartheta(x, u) - \lambda, & \text{otherwise}
\end{cases}
\end{equation}

We measure the similarity of $x$ and $u$, $\vartheta(x,u)$, as follows:
\begin{equation}\label{eq:virtual}
\resizebox{0.9\hsize}{!}{
	$\vartheta(x,u) = {\sum_{c \in CN(x,u)}(1-d(x,c) + 1-d(u,c))  \over { \sum_{ k\in \Phi(x)}(1-d(x,k)) + \sum_{l\in \Phi(u)} (1-d(u,l))   } } $
}
\end{equation}





After computing DI, CI, EI for each edge $(u,v)\in E$, new distance $d(u,v)$ at timestamp $t+1$ is updated as follows:
\begin{equation}
\label{eq:update_equation}
d^{t+1}(u,v) = d^t(u,v) - DI(u,v) - CI(u,v) - EI(u,v)
\end{equation}

Attractor algorithm is looped until every edge distance converged (e.g., its distance becomes either 0 or 1). Thirdly, Attractor removes edges with distance 1 and finds connected communities with breath first search. Each connected component is an identified community.

\section{MRAttractor}
\label{sec:MRAttractor}
\begin{figure}[t]
	\small
	\begin{tikzpicture}[node distance = 2cm, auto]
	\node [rounded rectangle,draw,fill=initstepcolor,text width=30mm,align=center] (init) {Jaccard Distance Initialization };

	\node [block, below of=init, text width=30mm,minimum height=1em,node distance=1.2cm,fill=dynamicinteractionscolor] (genstargraph) { Generating Star Graphs};
	
	\node [block, right of=genstargraph,minimum height=1em,text width=30mm,node distance = 4.5cm,fill=dynamicinteractionscolor] (dynamicinteraction) {Computing Three Types of Interaction};
	\node [block, below of=dynamicinteraction,text width=30mm,fill=dynamicinteractionscolor,node distance=1.7cm] (updateedge) {Updating Distances and sliding window};
	
	\node [decision, aspect=2, text width=5em, inner sep=-2pt, below of=genstargraph,node distance=1.7cm, rounded corners,fill=dynamicinteractionscolor] (decide) {\#non-converged edges $<\gamma$};
	\node [block, below of=decide,text width=30mm,node distance=2cm,fill=dynamicinteractionscolor] (single) {Processing on the Master node until all edges converged};
	\node [rounded rectangle,fill=boxcolor, right of=single,text width=30mm,align=center,draw,node distance=4.5cm] (stop) {Extracting communities};
	\path [line] (init) -- (genstargraph);
	\path [line] (genstargraph) -- (dynamicinteraction);
	\path [line] (decide) -- node {no} (genstargraph);
	\path [line] (dynamicinteraction) -- (updateedge);
	\path [line] (updateedge) -- (decide);
	\path [line] (decide) -- node {yes}(single);
	\path [line] (single) -- (stop);
	
	\end{tikzpicture}
	\vspace{-5pt}
	\caption{Flowchart of three major components MRAttractor. (1) Jaccard Distance Initialization, (2) Dynamic Interactions and (3) Extracting communities}
	\label{fig:MRAttractor_major_components}
	\vspace{-15pt}
\end{figure}
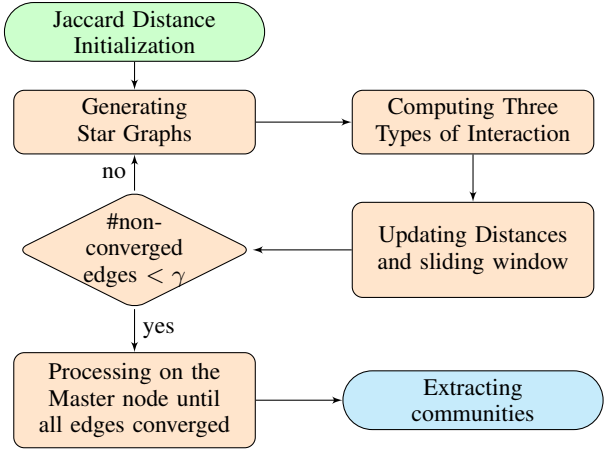
In this section, we describe \emph{MRAttractor}, our proposed distributed version of Attractor. It not only produces the same results with Attractor, but also significantly reduces the running time in both single machine and distributed system. 

MRAttractor consists of three main components. The first one is to initialize Jaccard distance. The second component is to compute dynamic interactions and make edge distance converged, which consists of four phases such as generating star graphs, computing three types of interaction, updating distances, and  running all edge convergence on the master node. The third component is to extract communities. Figure~\ref{fig:MRAttractor_major_components} shows three main components of MRAttractor. We explain each component in detail in the following subsections.


\subsection{Jaccard Distance Initialization}
For each vertex $u$, we find its neighbors $\Phi(u)$ and sort these neighbors increasingly based on their indexes. Then, for each edge $(u,v)$ of graph $G(V,E)$, we can find common neighbors of $u$ and $v$ and compute Jaccard distance (See Eq.\ref{eq:init-jaccardsim}) with complexity O($deg(u)+deg(v)$). 

\subsection{Computing Dynamic Interactions}
After initializing all edge distances of $G(V, E)$, we move on to the second major component of \textsc{MRAttractor} which consists of three MapReduce phases (i.e., generating star graphs, computing three types of Interactions, and updating distances based on sliding window), and running on the master node to make all edge distances converge. 
\begin{algorithm}[H]
	\caption{MR1: Generating Star Graphs}
	\small
	\label{alg:generating_star_graphs}
	\begin{algorithmic}[1]
		\Statex \textbf{Map:} \textbf{Input} $ \big \langle (u,v); d(u,v) \big \rangle $
		
		\State emit $\big \langle {u ; v\;d(u,v)} \big \rangle$ ; emit $\big \langle {v ; u\;d(u,v)} \big \rangle$
		\Statex \textbf{Reduce:} \textbf{Input}  $\big \langle u ; \Gamma(u) \big \rangle$
		\State Sort $\Gamma(u)$ increasingly based on index of $u$'s neighbors.
		\State emit $\big \langle u;\;Sorted(\Gamma(u))\big \rangle$
		
	\end{algorithmic}
\end{algorithm}
\subsubsection{Generating Star Graphs}
Algorithm~\ref{alg:generating_star_graphs} processes each edge $(u,v)$ and its distance. Then, in reduce step, we sort $\Gamma(u)$ based on index of $u$'s neighbors and output its star graph $\big \langle u;\Gamma(u)\big \rangle$. Note that a star graph is a tree of $k$ nodes where center vertex has degree $k-1$, while other vertices have degree 1. Sorting helps us find common and exclusive neighbors of two linked nodes in linear time. Totally, there will be $n=|V|$ star graphs output from reduce instances of Algorithm~\ref{alg:generating_star_graphs}.

\begin{figure*}
	\centering
	\includegraphics[trim=15cm 6.5cm 25.5cm 8cm,clip,width=0.98\textwidth,height=4cm]{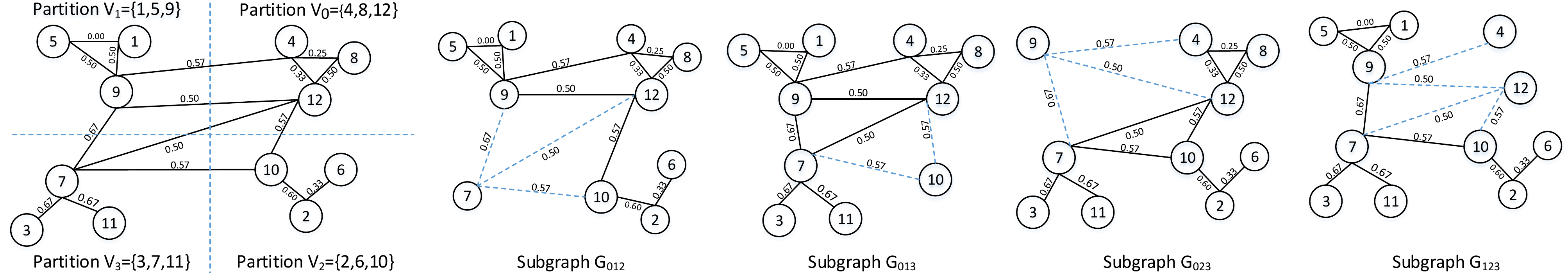}
	\vspace{-10pt}
	\caption{Original graph $G(V,E)$ where $|V|=12$, $|E|=16$. Each edge is associated with Jaccard distance in Equation \ref{eq:init-jaccardsim}. This graph is partitioned into four subgraph $G_{012}$, $G_{013}$, $G_{023}$ and $G_{123}$. The hash function $P(u)=u \mod p$ where $p=4$. The solid lines are main edges and the dot lines are rear edges. }
	\label{fig:big-picture-dynamic-interactions}
	\vspace{-10pt}
\end{figure*}

\subsubsection{Computing Three Types of Interactions}

Direct Interaction (DI), Common Interaction (CI) and Exclusive Interaction (EI) are three interactions we need to compute. The hardest task is computing $EI(u,v)$ of edge $(u,v)$ because $EI(u,v)$ depends on $\vartheta(x,u)$, the similarity of unconnected nodes $x$ and $u$ where $x\in EN(v)$ (see Eq.\ref{eq:EI-partial} and Eq.\ref{eq:rho}). Well-known large-scale graph processing frameworks \cite{gonzalez2014graphx,malewicz2010pregel,kang2009pegasus} limitedly support direct communication between two unconnected nodes (e.g. vertex $x$ and vertex $u$), leading to difficulties in computing $EI(u,v)$. Therefore, we propose \textsc{DecGP}, an algorithm to efficiently compute dynamic interactions of every edge $(u,v)$ in $G(V,E)$. Our proposed \textsc{DecGP} algorithm was inspired by a graph partitioning algorithm introduced in \cite{suri2011counting}, which showed that a graph partition algorithm helped mitigate skewness issue by evenly distributing workload to each reducer. The third challenge mentioned in Section~\ref{sec:intro} will be resolved by our graph partition algorithm.

Our graph partition algorithm uses a hash function $P(\textbf{.})$, that maps each vertex to range $[0,p-1]$, to partition original graph $G(V, E)$ into $p$ disjoint partitions $V_1, V_2,...,V_p$ such that $V=V_1 \cup V_2\cup...\cup V_p$, and $V_i \cap V_j=\emptyset$. From these partitions, we form overlapping subgraphs $G_{ijk}=(V_{ijk}, E_{ijk})$ where $V_{ijk}=V_i\cup V_j\cup V_k$, $E_{ijk}=\{(u,v)\in E| u,v \in V_{ijk} \}$. An edge $(u,v) \in E$ is called an outer-edge if $u$ and $v$ are in different partition (i.e., $P(u) \neq P(v)$). Otherwise, it is an inner-edge. Intuitively, in each smaller-size subgraph $G_{ijk}$, we compute DI, CI and EI of edges $(u,v)\in E_{ijk}$. Let's take an example graph $G(V,E)$ of 12 vertices and 16 edges in Figure \ref{fig:big-picture-dynamic-interactions}. By using hashing function $P(u)=u\mod p$ where $p=4$, this graph is partitioned into 4 subgraphs $G_{012}$, $G_{013}$, $G_{023}$ and $G_{123}$. The value on each edge is Jaccard distance. Although our algorithm shares similar methodology with \cite{suri2011counting}, there are key differences as follows:

\squishlist
\item We reduce complexity of finding subgraphs $G_{ijk}$ that contain an edge $(u,v)$ from O($p^3$) to O($p^2$).
\item In our algorithm, we also include additional edges called \textbf{rear edges} in subgraphs $G_{ijk}$ to compute exclusive interaction while in \cite{suri2011counting}, subgraphs $G_{ijk}$ only contain \textbf{main edges}. In each subgraph $G_{ijk}$, an edge $(u,v)$ is called a main edge if $\{P(u),P(v)\} \in \{i,j,k\}$ otherwise it is called a rear-edge. In Figure~\ref{fig:big-picture-dynamic-interactions}, subgraphs $G_{012}$, $G_{013}$, $G_{023}$ and $G_{123}$ contain rear-edges denoted as dotted edges and main edges denoted as solid edges.
\item In each subgraph $G_{ijk}$, we compute partial values of $DI(u,v)$, $CI(u,v)$ and $EI(u,v)$ of every main edges $(u,v)$ instead of counting the number of triangles like \cite{suri2011counting}.
\squishend

These differences are described as follows:

\noindent\textbf{(i) Reducing complexity of finding subgraphs ${G_{ijk}}$.}

Given an edge $(u,v)$, a graph partition algorithm of Siddharth et al. \cite{suri2011counting} takes $O(p^3)$ for finding subgraphs $G_{ijk}$ that $(u,v)$ belongs to. However, our proposed algorithm only takes $O(p^2)$ for the task (see Lemma \ref{lm:mrgp_www_inner_outer} and associated proof).


\begin{lemma}
	For each edge $(u,v)$ in original graph $G(V,E)$:
	\begin{enumerate}
		\item If edge $(u,v)$ is an inner edge, there will be ${(p-1)\cdot (p-2)} \over 2$ distinct subgraphs $G_{ijk}$ containing it.
		\item If edge $(u,v)$ is an outer edge, there will be $p-2$ subgraphs $G_{ijk}$ containing this edge.
	\end{enumerate}
	\label{lm:mrgp_www_inner_outer}
\end{lemma}

\begin{proof}
	(1) If edge $(u,v)$ is an inner edge (e.g., $P(u)=P(v)=i$), it will be emitted to $G_{ijk}$ for each $j\in [0;p-1],\,k\in [0;p-1]$ and $j \neq i$, $i\neq k$, $k\neq j$. Therefore, there are ${(p-1)\cdot (p-2)} \over 2$ subgraphs $G_{ijk}$ containing the inner edge $(u,v)$. 
	
	(2) If edge $(u,v)$ is an outer edge (e.g., $P(u)=i,P(v)=j$), it will be output to $G_{ijk}$ where $i=P(u)$ and $j=P(v)$ for each $k\in [0;p-1], k\neq i$ and $k\neq j$. So, there are $p-2$ subgraphs $G_{ijk}$ containing this edge.
\end{proof}
Due to Lemma \ref{lm:mrgp_www_inner_outer}, finding subgraphs $G_{ijk}$ that contain an edge $(u,v)$ can be implemented in quadric complexity as shown in Algorithm \ref{alg:improved-find-subgraphs-G_ijk}.

\begin{algorithm}[H]
	\caption{Finding $G_{ijk}$ containing an edge $(u,v)$}
	\begin{algorithmic}[1]
		\small
		\Function{FindSubGraphs}{$u,v$}
		\State $S=\emptyset$
		\If{$P(u)=P(v)$}
		\For{$a \in [0;p-1]$ \textbf{and} $a\neq P(u)$}
		\For{$b \in [a+1;p-1]$ \textbf{and} $b\neq P(u)$}
		\State $S = S \cup \{sorted(a,b,P(u)) \}$
		\EndFor
		\EndFor
		\Else
		\For{$a \in [0;p-1]$ \textbf{and} $a\neq P(u)$ \textbf{and} $a\neq P(v)$}
		\State $S = S \cup \{sorted(a,P(u),P(v)) \}$
		\EndFor
		\EndIf
		\Return $S$
		\EndFunction
		
	\end{algorithmic}
	\label{alg:improved-find-subgraphs-G_ijk}
\end{algorithm}

\noindent\textbf{(ii) Adding rear-edges to subgraphs $G_{ijk}$.} What is the motivation of adding rear-edges?

Let's consider only main edges $(u,v)$ of subgraphs $G_{ijk}$. Once again, an edge $(u,v)$ in subgraph $G_{ijk}$ is a main edge if $\{P(u), P(v)\} \in \{i,j,k\}$. Otherwise, it is a rear-edge. For each main edge $(u,v)$ associated with its distance in each subgraph $G_{ijk}$, we can load vertices' degree into memory, and compute $DI(u,v)$ based on Eq.\ref{eq:DI-1} and $CI_c(u,v)$ for every common vertex $c$ based on Eq.\ref{eq:CI}. But how can we compute $EI(u,v)$? Let's look at a main edge (12, 10) of subgraph $G_{012}$ in Figure \ref{fig:big-picture-dynamic-interactions} as an example. We can see that vertex 9 is an exclusive neighbor of vertex 12. $EI(10,12)$ depends on $EI_9(10,12)$ (See Equation \ref{eq:EI-partial}). To compute $EI_9(10,12)$, we need to measure similarity $\vartheta(9,10)$ in Eq.\ref{eq:virtual} and $\rho(9,10)$ in Eq.\ref{eq:rho}. To accurately compute $\vartheta(9,10)$, we need both star graph $(9, \Gamma(9))$ and star graph $(10, \Gamma(10))$ to find common neighbors. But in subgraph $G_{012}$, without considering rear-edges, we are missing necessary edges $(9,7)$ and $(10,7)$ to compute $\vartheta(9,10)$ since $P(7)=3 \notin \{0,1,2\}$. Similar analysis to vertex 2, an exclusive neighbor of vertex 10, we are missing edge $(12,7)$ to compute $\vartheta(12,2)$. 
Motivated by these observations and to guarantee the correctness of MRAttractor, we added rear edges $(9,7)$, $(10,7)$ and $(12,7)$, denoted as dotted edges, to subgraph $G_{012}$ to correctly compute exclusive interactions. We resolved the first challenge mentioned in Section~\ref{sec:intro} by adding rear-edges.
\begin{algorithm}[H]
	\caption{MR2: \textsc{DecGP} - Computing 3 Types Interactions}
	\small
	\begin{algorithmic}[1]
		\Statex \textbf{Map:} \textbf{Input}: $\big \langle {u; \Gamma(u)} \big \rangle$
		\State $S_{T}=\emptyset$ \Comment Set of triples $(i,j,k)$
		\For{\textbf{each} $(v, d(u,v))\in \Gamma(u)$ }
		\State $S_{T} = S_T\;\cup $ \Call{FindSubGraphs}{$(u,v)$}
		\EndFor
		
		\For{\textbf{each} $(i,j,k)\in $  $S_T$}
		\State emit $\big \langle (i,j,k);(u, \Gamma(u)) \big \rangle$
		\EndFor

		\Statex \textbf{Reduce:} \textbf{Input:} $ \big \langle (i,j,k); \{(u,\Gamma(u))|P(u) \in \{i,j,k\} \} \big \rangle $ 
		\State Read value $p,\lambda$ from Hadoop configuration object.
		\State $stars=\emptyset$ \Comment An empty dictionary of star graphs
		\State $S_M = \emptyset$ \Comment Set of main edges.
		\For{\textbf{each} $(u,\Gamma(u))$ routed by key $(i,j,k)$}
		\State $stars[u]=\Gamma(u)$ \Comment Insert to dictionary
		\For{each $(v,d(u,v)) \in \Gamma(u)$}
		\If{$(P(u),P(v)) \in \{i,j,k\}$} \Comment $(u,v)$ is a main edge
		\State $S_M = S_M \cup \{sorted(u,v),d(u,v) \}$
		\EndIf
		\EndFor
		\EndFor
		\For{\textbf{each} $(u,v,d(u,v)) \in S_M $ \textbf{and} $0<d(u,v)<1$}
		\State $S_I=0$
		\State $S_I+=$\Call{ComputeDI}{$u,v,d(u,v),p$}
		\State $S_I+=$\Call{ComputeCI}{$u,v,d(u,v),\Gamma(u),\Gamma(v),p,\{i,j,k\}$}

		\For{ $(x, d(v,x)) \in \Gamma(v)$ and $x\in EN(v)$}
		\If{$P(x) \in \{i,j,k\}$}
		\State $\Gamma(x) = stars[x]$ \Comment Retrieve neighbors of vertex $x$
		\State $S_I+=$\Call{ComputeEI}{$\wedge(u,v,x),\Gamma(u),\Gamma(x),p,\lambda$}
		\EndIf
		\EndFor
		
		\For{$(y, d(u,y)) \in \Gamma(u)$ and $y\in EN(u)$}
		\If{$P(y)\in \{i,j,k\}$}
		\State $\Gamma(y) = stars[y]$ \Comment Retrieve neighbors of vertex $y$
		\State $S_I+=$\Call{ComputeEI}{$\wedge(v,u,y),\Gamma(v),\Gamma(y),p,\lambda$}
		\EndIf
		\EndFor
		\State emit $\big \langle {(u,v); S_I} \big \rangle$
		\EndFor
	\end{algorithmic}
	\label{alg:compute_three_interactions}
	
\end{algorithm}
\noindent\textbf{(iii) Computing DI, CI and EI of main edges.}
By adding rear edges to $G_{ijk}$, we are now able to compute DI, CI and EI of main edges in subgraph $G_{ijk}$. Algorithm~\ref{alg:compute_three_interactions} shows the pseudocode of \textsc{DecGP} to compute DI,CI and EI of every main edge in $G_{ijk}$. It is a MapReduce algorithm consisting of a map function and a reduce function. In particular, each map instance handles a star graph $(u, \Gamma(u))$ output from Algorithm \ref{alg:generating_star_graphs}. It firstly finds distinct subgraphs $G_{ijk}$, represented as a sorted triple $(i,j,k)$, which contains at least one of the edges $(u,v, d(u,v))$ where $(v,d(u,v)) \in \Gamma(u)$ by applying function \textsc{FindSubGraphs}. 
Secondly, it emits star graph $(u, \Gamma(u))$ to subgraph $G_{ijk}$ (See lines [2-5] of \textsc{DecGP}). By emitting the whole star graph $(u, \Gamma(u))$ to subgraph $G_{ijk}$, we ensure that there are always enough edges to compute similarity of unconnected nodes.

Moving on to each reduce instance of Algorithm \ref{alg:compute_three_interactions}. Each reduce instance receives a list of star graph $(u, \Gamma(u))$ routed by the sorted key $(i,j,k)$ of subgraph $G_{ijk}$, which can be viewed as an adjacent list representation of subgraph $G_{ijk}$. From lines [8-13] in Algorithm~\ref{alg:compute_three_interactions}, we find set $S_M$ of main edges $(u,v, d(u,v)) \in E_{ijk}$, which are used for computing DI, CI and EI in Lines [14-26]. After computing DI, CI and EI of every non-converged main edge in $G_{ijk}$, each reduce instance will emit key-value pairs where key is a main edge $(u,v)$ and value is the aggregated DI, CI and EI of edge $(u,v)$. Three main functions \textsc{ComputeDI}, \textsc{ComputeCI} and \textsc{ComputeEI} are explained below. 

\textbf{Computing DI.} Algorithm \ref{alg:computeDI} shows how we compute Direct Interaction. $DI(u,v)$ of each main edge $(u,v)$ is computed based on Eq.\ref{eq:DI-1}. However, this edge $(u,v)$ will be re-computed in multiple subgraphs $G_{ijk}$ based on Lemma \ref{lm:mrgp_www_inner_outer} so we need to scale $DI(u,v)$ down to guarantee correctness. In particular, if $(u,v)$ is an inner edge, we scale $DI(u,v)$ down ${(p-1)\cdot (p-2)}\over 2$ times. Otherwise, we scale it down $p-2$ times.
\begin{algorithm}[H]
	\caption{Direct interaction of the edge $(u,v)$}
	\small
	\begin{algorithmic}[1]
		\Function{ComputeDI}{$u,v,d(u,v),p$}
		\State $DI={sin(1 - d(u,v)) \over {deg(u)} } + {sin(1 - d(u,v)) \over {deg(v)}}$
		\If{$P(u)=P(v)$}
		\State $DI = {{2\cdot DI} / { ((p-1)\cdot (p-2)) } } $
		\Else
		\State $DI = DI / { (p-2) }$
		\EndIf
		\Return $DI$
		\EndFunction
		
	\end{algorithmic}
	\label{alg:computeDI}
\end{algorithm}

\textbf{Computing CI.} Algorithm \ref{alg:computeCI} shows how we compute common interaction $CI(u,v)$ of main edge $(u,v)$. The input of this function is the main edge $(u,v)$, its distance $d(u,v)$, $u$'s neighbors $\Gamma(u)$, $v$'s neighbors $\Gamma(v)$, the number of partitions $p$ and key of subgraph $G_{ijk}$. For each common neighbor $c$ of $u$ and $v$, we only consider vertex $P(c)\in \{i,j,k\}$ because we only care about main edges of $G_{ijk}$. We compute $CI_c(u,v)$ based on Eq.\ref{eq:CI}. However, we need to scale down $CI_c(u,v)$ since three vertices $u$,$v$ and $c$ will form a triangle $\triangle(u,v,c)$ and Siddharth et al., \cite{suri2011counting} pointed out that $\triangle(u,v,c)$ can be repeated in several subgraphs $G_{ijk}$ in Lemma \ref{lm:www_repeatedness_triangle_lemma}. In particular, if three vertices of $\triangle(u,v,c)$ are in a same partition $V_i$, we scale $CI_c(u,v)$ down ${(p-1)\cdot (p-2)}\over {2}$ times. If $\triangle(u,v,c)$ has two nodes in same partition and the third node belongs to a different partition, $CI_c(u,v)$ is scaled down $p-2$ times.
\begin{algorithm}[H]
	\caption{Common interaction of the edge $(u,v)$}
	\small
	\begin{algorithmic}[1]
		\Function{ComputeCI}{$u,v,d(u,v),\Gamma(u),\Gamma(v),p, \{i,j,k\}$}
		\State $CI=0$
		\For{$c\in CN(u,v)$ \textbf{and} $P(c) \in \{i,j,k\}$ \textbf{and} $(c,d(u,c)) \in \Gamma(u)$ \textbf{and} $(c,d(v,c)) \in \Gamma(v)$}
		\State $w_1 = 1 - d(u,c)$ ; $w_2 = 1 - d(v,c)$
		\State $CI_c(u,v)= {{w_2 \cdot sin(w_1)} \over {deg(u)} } + { {w_1 \cdot sin(w_2)} \over {deg(v)} }$
		\If{$P(u) = P(v) = P(c)$}
		\State $CI_c(u,v) = {2\cdot CI_c(u,v) / {((p-1)\cdot (p-2))} }$
		\ElsIf{\small{$P(u)=P(v)\,|\,P(u)=P(c)\,|\,P(v)=P(c)$}}
		\State $CI_c(u,v) = {CI_c(u,v)/ ({p-2})}$
		\EndIf
		\State $CI = CI + CI_c(u,v)$
		\EndFor
		\State \textbf{return} $CI$
		\EndFunction
	\end{algorithmic}
	\label{alg:computeCI}
\end{algorithm}

\begin{lemma}
	\label{lm:www_repeatedness_triangle_lemma}
	Given $\triangle(u,v,c)$ of original graph $G(V,E)$:
	\begin{enumerate}
		\item if three nodes are in a same partition, $\triangle(u,v,c)$ will also appear in ${(p-1)\cdot (p-2)}\over {2}$ different subgraphs $G_{ijk}$.
		\item if two nodes are in a same partition and one node belongs to a different partition, $\triangle(u,v,c)$ will appear in $p-2$ different subgraphs $G_{ijk}$.
		\item if three nodes are in three different partition, there is only one subgraph containing $\triangle(u,v,c)$.
	\end{enumerate}
\end{lemma}
\begin{proof}
	(1) When $u$, $v$ and $c$ are in a same partition $P(u)=P(v)=P(c)=i$, $\triangle(u,v,c)$ will appear in subgraphs $G_{ijk}$ where $j\in [0;p-1]$, $k\in [0;p-1]$ and $j\neq k, k\neq i, i\neq j$. Totally, there are ${(p-1)\cdot (p-2)}\over {2}$ different subgraphs $G_{ijk}$.
	
	(2) Without loss of generality, we assume that $P(u)=P(v)=i$ and $P(c)=j$. $\triangle(u,v,c)$ will appear in subgraphs $G_{ijk}$ where $k\in [0;p-1]$ and $k\neq i$ and $k \neq j$. Therefore, there are $p-2$ distinct subgraphs $G_{ijk}$ containing $\triangle(u,v,c)$.
	
	(3) Finally, when three nodes are in three different partitions $P(u)=i$, $P(v)=j$ and $P(c)=k$, there is only one subgraph $G_{ijk}$ that contains this $\triangle(u,v,c)$.
\end{proof}
\begin{algorithm}[H]
	\caption{Exclusive interaction of vertex $x$ on edge $(u,v)$}
	\begin{algorithmic}[1]
		\Function{ComputeEI}{$\wedge(u,v,x),\Gamma(u),\Gamma(x),p,\lambda$}
		\State $EI_x(u,v)=0$
		\State Compute $\vartheta(x,u)$ based on $\Gamma(u)$ and $\Gamma(x)$
		\State  $\rho(x,u) = \vartheta(x,u)$
		\If{$\vartheta(x,u) < \lambda$}
		\State $\rho(x,u) = \vartheta(x,u)  - \lambda$
		\EndIf
		\State $EI_x(u,v) = \rho(x,u) \cdot sin(1 - d(v,x)) / deg(v)$
		\If{$P(u) = P(v) = P(x)$}
		\State $EI_x(u,v) = {2\cdot EI_x(u,v) / {((p-1)\cdot (p-2))} }$
		\ElsIf{\small{$P(u)=P(v)\,|\,P(u)=P(x)\,|\,P(v)=P(x)$}}
		\State $EI_x(u,v) = {EI_x(u,v)/ ({p-2})}$
		\EndIf
		
		\State \textbf{return} $EI_x(u,v)$
		\EndFunction
	\end{algorithmic}
	\label{alg:computeEI}
\end{algorithm}
\textbf{Computing EI.} For each main edge $(u,v)$ of $G_{ijk}$, we sequentially process each exclusive neighbor $x\in EN(v)$ (See Lines [18-21] of Algorithm \ref{alg:compute_three_interactions}) and each exclusive neighbor $y\in EN(u)$ (See Lines [22-25] of Algorithm \ref{alg:compute_three_interactions}). Once again, we only pay attention to vertex $x$ and vertex $y$ such that $P(x)\in \{i,j,k\}$ and $P(y)\in \{i,j,k\}$ since main edges are our target.

Before describing details of the computation, we present definition of a wedge or a two-hop path in Definition \ref{def:wedge}.
\begin{mydef}[Wedge]{\ \\}
	In $G(V,E)$, three nodes $u$, $v$ and $x$ form a wedge denoted as $\wedge(u,v,x)$ if $(u,v) \in E, (v,x) \in E$ and there is no edge between $u$ and $x$. A wedge can be called two-hop path \cite{gupta2016decompositions}.
	\label{def:wedge}
\end{mydef}

As we can see, three nodes $u,v$ and $x$ form a $\wedge(u,v,x)$. Similarly, three nodes $v,u$ and $y$ create $\wedge(v,u,y)$. Without loss of generality, we only explain how we compute $EI_x(u,v)$, the effect of exclusive neighbor $x\in EN(v)$ on distance of edge $(u,v)$ based on $\wedge(u,v,x)$.

In function \textsc{computeEI} of Algorithm \ref{alg:computeEI}, we input $\wedge(u,v,x)$, $u$'s neighbors, $x$'s neighbors, the number of partitions $p$ and cohesive parameters $\lambda$ \cite{kdd.dyna}. To begin with, we compute $\vartheta(x,u)$ (See Eq.\ref{eq:virtual}), the similarity of vertex $u$ and $x$ based on $\Gamma(u)$ and $\Gamma(x)$. Then, we derive $\rho(x,u)$ and compute $EI_x(u,v)$ in Eq.\ref{eq:EI-partial}. Computing $EI_x(u,v)$ in each subgraph $G_{ijk}$ will face duplication problem since wedge $\wedge(u,v,x)$ can appear in other different subgraphs $G_{ijk}$. Therefore, we need to scale down $EI_x(u,v)$ appropriately. Lemma \ref{lm:lm_wedge} shows the number of subgraphs $G_{ijk}$ that a wedge $\wedge(u,v,x)$ can appear. In particular, if $\wedge(u,v,x)$ has three nodes in a same partition, we scale $EI_x(u,v)$ down ${(p-1)\cdot (p-2)}\over 2$ times. If the first two nodes in $\wedge(u,v,x)$ are in the same partition, but the third node is in another partition, we scale $EI_x(u,v)$ down $p-2$ times (see Lines [8-11] of Algorithm \ref{alg:computeEI}).

\begin{lemma}\label{lm:lm_wedge} For each wedge $\wedge(u,v,x)$:
	\begin{enumerate}
		\item If three nodes $u$, $v$ and $x$ are placed in the same partition, it will appear in ${(p-1)\cdot (p-2)}\over 2$ different subgraphs $G_{ijk}$.
		\item If two nodes are in the same partition and the other one belongs to a different partition, it will appear in $p-2$ different subgraphs $G_{ijk}$.
		\item If each vertex is in different partitions, it will belong to only one subgraph $G_{ijk}$.
	\end{enumerate}
\end{lemma}
\begin{proof}
	(1) When three nodes are in a same partition, $P(u)=P(v)=P(w)=i$. Edges $(u,v)$ and $(v,x)$ are two inner edges and always appear together. Due to Lemma \ref{lm:mrgp_www_inner_outer}, inner edges will appear in ${(p-1)\cdot (p-2)}\over 2$ subgraphs $G_{ijk}$. Therefore, $\wedge(u,v,x)$ will appear in ${(p-1)\cdot (p-2)}\over 2$ different subgraphs $G_{ijk}$.
	
	(2) Without loss of generality, we assume that $P(u)=P(v)=i$ and $P(x)=j$. For each subgraph $G_{ijk}$ that outer edge $(v,x)$ appears, we can see that inner edge $(u,v)$ also appears, leading to the existence of $\wedge(u,v,x)$ in $G_{ijk}$. We know that the outer edge $(v,x)$ will appear in $p-2$ different subgraphs $G_{ijk}$ based on Lemma \ref{lm:mrgp_www_inner_outer}. Therefore, $\wedge(u,v,x)$ will appear in $p-2$ different subgraphs $G_{ijk}$.
	
	(3) If each vertex are in different partitions, $P(u)=i$, $P(v)=j$, $P(x)=k$. Thus, only one $G_{ijk}$ has this wedge.
\end{proof}

\subsubsection{Updating edge distances based on sliding window}
Next, we move on to updating distance of every edge $(u,v)$ and its sliding window in original graph $G(V,E)$ based on $DI(u,v)$, $CI(u,v)$ and $EI(u,v)$ computed by Algorithm~\ref{alg:compute_three_interactions}.

In Section~\ref{sec:intro}, we addressed three key challenges to design and implement MRAttractor. The second challenge was how to make all edge distances converged with minimum overhead of network communication and disk I/O. To overcome the second challenge, we use sliding window technique \cite{chang2003finding}.

\begin{algorithm}[!ht]
	\caption{MR3: Updating Distances and Sliding Window}\label{alg:updateEdge}
	
	\begin{algorithmic}[1]
		\Statex \textbf{Map:} \textbf{Input:} \small$\big \langle (u,v); S_I \big \rangle$, $\big  \langle(u,v);d(u,v)\big \rangle$ and $\big  \langle(u,v);\textbf{w}\big \rangle$
		\State emit $\big \langle (u,v); vl \big \rangle$ with $vl$ is either $S_I$, $d(u,v)$ or $\textbf{w}$
		\Statex \textbf{Reduce:} \textbf{Input:} $\big \langle (u,v); values \big \rangle$
		\State Read settings $s$ the maximum size of sliding window and $\tau$
		\State $d^{t}(u,v)=0$; $\Delta^t(u,v)=0$; $\textbf{w}^{t+1}=\emptyset$
		\For{$vl \in values$}
		\If{$vl$ \textbf{is} $\textbf{w}$}
		\State $\textbf{w}^{t+1}=\textbf{w}$
		\ElsIf{$vl$ is $d(u,v)$}
		\State $d^{t}(u,v) \leftarrow {d}(u,v)$
		\ElsIf{$vl$ is $S_I$}
		\State $\Delta^t(u,v) \leftarrow \Delta^t(u,v)  + S_I $
		\EndIf
		\EndFor
		
		\If{$ d^t(u,v)=0$ \textbf{or} $d^t(u,v)=1$ }
		\State \textbf{return} \Comment This edge was converged. No need to process it
		\EndIf
		
		\If{$ \Delta^t(u,v) \neq 0$}
		\State $d^{t+1}(u,v)$ = $d^{t}(u,v) -  \Delta^t(u,v)$
		\State $index=(t+1)\mod s$
		\State $\textbf{w}^{t+1}[index]=-1$ \Comment Set position $index$ of vector $\textbf{w}^{t+1}$
		\If{$d^{t+1}(u,v) > d^{t}(u,v)$}
		\State $\textbf{w}^{t+1}[index]=1$
		\EndIf
		\State $x$ is the number of 1 in sliding window $\textbf{w}^{t+1}$
		\State $y$ is the number of -1 in sliding window $\textbf{w}^{t+1}$
		\If{$t+1 \geq s$} \Comment $\textbf{w}^{t+1}$ of edge $(u,v)$ is full
		\If{$\textbf{w}^{t+1}[index]=1$ \textbf{and} $x\geq \tau \cdot s$}
		\State $d^{t+1}(u,v)=1$
		\EndIf
		\If{$\textbf{w}^{t+1}[index]=-1$ \textbf{and} $y\geq \tau \cdot s$}
		\State $d^{t+1}(u,v)=0$
		\EndIf
		\EndIf
		
		\If{$d^{t+1}(u,v) \geq 1$}
		\State $d^{t+1}(u,v) = 1$
		\EndIf
		\If{$d^{t+1}(u,v)\leq 0$}
		\State $d^{t+1}(u,v) = 0$
		\EndIf
		
		\State emit $\big \langle (u,v);d^{t+1}(u,v) \big \rangle$
		\State emit $\big \langle(u,v);\textbf{w}^{t+1} \big \rangle$
		\EndIf
	\end{algorithmic}
\end{algorithm}

Our sliding window model works as follows:
For each edge $(u,v)$, we use a binary vector $\textbf{w}$, indicating the status of edge $(u,v)$, $\textbf{w}_i = 1$ means that $d(u,v)$ at iteration $i^{th}$ increases and $\textbf{w}_i = -1$ means $d(u,v)$ decreases. By using sliding window, we only keep the last $s$ statuses of each edge (i.e. $\textbf{w}\in \textbf{R}^s$) to observe the increasing/decreasing trend of the edge $(u,v)$ in the last $s$ iterations which may be more reliable to reflect the convergence trend of an edge. Then, we predict its final distance (i.e. 0 or 1). To decide if an edge converges or not, we use a threshold $\tau \in [0;1]$. In particular, if the last status of edge $(u,v)$ is $-1$ and there are at least $\tau \times s$ negative values in vector $\textbf{w}$, we decide that edge will eventually converge to 0 (e.g., when $s$=10, $\tau$=0.6 and last status=-1, if at least 6 statuses in $\textbf{w}$ are -1, edge distance will be set to 0). If the last status of edge $(u,v)$ is $1$ and there are at least $\tau \times s$ positive values in vector $\textbf{w}$, we will set $d(u,v)=1$. Otherwise, edge $(u,v)$ still does not converge, and we continue to compute its dynamic interactions as shown in Figure~\ref{fig:MRAttractor_major_components}.




Algorithm~\ref{alg:updateEdge} shows our pseudocode to update edge distances. The map instances of Algorithm \ref{alg:updateEdge} process three types of input: (1) $\big \langle (u,v); S_I \big \rangle$ is the key-value pairs generated Dynamic Interactions. Recall that $S_I$ is the aggregated sum of DI, CI and EI of edge $(u,v)$ output from reduce instances of Algorithm~\ref{alg:compute_three_interactions}. (2) $\big  \langle(u,v);d(u,v)\big \rangle$ is the edge $(u,v)$ and its distance in previous iteration. At the first iteration, $d(u,v)$ is Jaccard distance. (3) $\big  \langle(u,v);\textbf{w}\big \rangle$ is sliding window of edge $(u,v)$. At the first iteration, $\textbf{w}$ is an empty vector. Note that, for each edge $(u,v)$, there are only one pair $\big  \langle(u,v);d(u,v)\big \rangle$, one pair $\big  \langle(u,v);\textbf{w}\big \rangle$ and multiple pairs $\big \langle (u,v); S_I \big \rangle$.

Each map instance of Algorithm \ref{alg:updateEdge} simply outputs a key-value pair where key is an edge $(u,v)$ and value is either $S_I$, $d(u,v)$ or sliding window vector $\textbf{w}$. 

Each reduce instance of Algorithm \ref{alg:updateEdge} receives $values$ routed by key $(u,v)$ and performs two tasks with edge $(u,v)$- (1) computing its new distance and (2) updating its sliding window vector. To begin with, from Lines [3-10], we sum up all $S_I$ values and store it into $\Delta^t(u,v)$. We can verify that $\Delta^t(u,v)=DI(u,v) + CI(u,v) + EI(u,v)$. After computing $\Delta^t(u,v)$, we can derive $d^{t+1}(u,v)$, distance of edge $(u,v)$ at timestamp $t+1$ based on Eq.\ref{eq:update_equation}. Next, from Lines [15-25], we update sliding window vector $\textbf{w}^{t+1}$. Due to modulo function, $index\in[0;s-1]$. Finally, we emit pair $\big \langle (u,v);d^{t+1}(u,v) \big \rangle$ for new distance of edge $(u,v)$ and pair $\big \langle(u,v);\textbf{w}^{t+1} \big \rangle$ for new sliding window vector. These key-value pairs will act as new input for next iteration of \textsc{MRAttractor}.

\subsubsection{Running on Master node}
After deriving new distance of all edges in original graph $G(V,E)$, we will check how many edges $(u,v)$ are still non-converged 0$<d^{t+1}(u,v)<$1. If the number of non-converged edges are smaller than a threshold $\gamma$ we will continue our computation on Master node, which control slave nodes. There are two reasons why we do this. Firstly, Attractor algorithm suffers long-tail iterations because some edges converge slowly \cite{meng2016improved}. Secondly, after multiple iterations, the number of non-converged edges are very small which can be handle efficiently on single computer. A well-known problem of MapReduce Hadoop is the overhead of I/O operations \cite{jiang2010performance}. Therefore, by running on single computer, we can avoid unnecessary overhead of Hadoop framework. In this work, we set $\gamma=10,000$ for all testing networks. The second challenge mentioned in Section~\ref{sec:intro} is resolved by the sliding window model and running on the master node. 

\subsubsection{Complexity Analysis}
 We now show the correctness of computing dynamic interactions and analysis of \textsc{DecGP}'s complexity since it is the most time consuming part.
\begin{lemma}
	\label{lm:correctness_DI_CI_EI}
	For each edge $(u,v)$ of $G(V,E)$, its $DI(u,v)$, $CI(u,v)$ and $EI(u,v)$ are computed correctly in each loop.
\end{lemma}
\begin{proof}
In each $G_{ijk}$, we compute partial values of DI, CI and EI for every main edge with appropriate scaling as shown in Algorithms \ref{alg:compute_three_interactions}, \ref{alg:computeDI}, \ref{alg:computeCI} and \ref{alg:computeEI}. After computing dynamic interactions, we aggregate DI, CI and EI for every edge of the original graph $G(V,E)$ in reduce instances of Algorithm \ref{alg:updateEdge}. Since we apply scaling correctly, the aggregated values $\Delta^{t+1}(u,v)$ of every edge is exactly equal to $DI(u,v)+CI(u,v)+EI(u,v)$, leading to correctness of our computation. 
\end{proof}

\begin{table}[!ht]
	\centering
	\resizebox{1.0\linewidth}{!}
	{
		\begin{tabular}{|l|l|l|l|l|l|}
			\hline
			Datasets      & $|$V$|$   & $|$E$|$   & $|$classes$|$                 & AVD    & CC    \\ \hline
			Karate        & 34        & 78        & 2                       & 4.588  & 0.571 \\ \hline
			Football      & 115       & 613       & 12                      & 10.661 & 0.403 \\ \hline
			Polbooks      & 105       & 441       & 3                       & 8.400  & 0.488 \\ \hline
			Amazon        & 334,863   & 925,872   & top5000                 & 5.530  & 0.397 \\ \hline
			Collaboration & 9,875     & 25,973    & \multicolumn{1}{c|}{unknown} & 5.260  & 0.472 \\ \hline
			Friendship    & 58,228    & 214,078   & \multicolumn{1}{c|}{unknown} & 7.353  & 0.172 \\ \hline
			Road          & 1,088,092 & 1,541,898 & \multicolumn{1}{c|}{unknown} & 2.834  & 0.046 \\ \hline
		\end{tabular}
	}
	\vspace{2pt}
	\caption{Networks with labels and non-labels, average degree (AVD), and average clustering coefficient (CC).}
	\label{tbl:sliding_window_datasets}
	\vspace{-15pt}
\end{table}

%
%
\begin{lemma}
	For each setting of $p$:
	\begin{enumerate}
			\item The expected number of main edges in $G_{ijk}$ is O($m\over p^2$).
			\item The expected number of key-value pairs $\big \langle (u,v); S_I \big \rangle$ is O($mp$) for all reduce instances of Algorithm \ref{alg:compute_three_interactions}.
	\end{enumerate}
	\label{lm:complexity_analysis}
\end{lemma}
\begin{proof}
	(1) The probability that $P(u) \in \{i,j,k\}$ is $O({3\over p})$. An edge $(u,v)$ in $G_{ijk}$ is a main edge if  $P(u)\in \{i,j,k\}$ and $P(v)\in \{i,j,k\}$. Therefore, the likelihood that an edge appears between $u$ and $v$ is $9\over p^2$, resulting in the expected number of main edges of $G_{ijk}$ is $O({m\over p^2})$, where $m$ is total number of edges in the original graph $G(V,E)$.
	
	(2) For each main edge $(u,v)$ of $G_{ijk}$, after computing partial value of DI, CI and EI, we emit one key-value pair $\big \langle (u,v); S_I \big \rangle $. Since the number of main edges is $O({m \over p^2})$ and the number of subgraphs $G_{ijk}$ is $O(p^3)$, the expected number of key-value pairs $\big \langle (u,v); S_I \big \rangle$ is $O(p^3 \cdot {m \over p^2}) = O(mp)$.
\end{proof}
Based on Lemma \ref{lm:complexity_analysis}, we approximately estimate complexity of computing dynamic interactions in each subgraph $G_{ijk}$.
For each main edge $(u,v)$ of $G_{ijk}$, time complexity of computing direct interaction is O(1). Time complexity to compute common interaction is O($deg(u)+deg(v)$). Time complexity of computing exclusive interaction is about O($T\cdot(deg(u)+deg(v))$) where T is average number of exclusive neighbors of each node. Therefore, for all main edges of $G_{ijk}$, the total complexity is about $T\cdot\sum_{(u,v)\in E_{ijk}}(deg(u)+deg(v) \leq T\cdot \sum_{u\in V_{ijk}}deg^2(u)$. It has been proved in \cite{de1998upper} that $\sum_{u\in V}deg^2(u) \leq m \big( { {2m} \over {n - 1} } + n - 2 \big)$ for the original graph $G(V,E)$. Then, what is the upper bound of $\sum_{u\in V_{ijk}}deg^2(u)$ for the subgraph $G_{ijk}$ which is smaller than $G(V,E)$? The expected number of vertices of $G_{ijk}$ is $O(n/p)$, so we can roughly estimate the upper bound of $T\cdot\sum_{u\in V_{ijk}}deg^2(u)$ as $O\big({mnT\over p} \big( { {2m} \over {n - 1} } + n - 2 \big)\big)$.

\begin{table*}[t]
	\centering
	\resizebox{1.0\linewidth}{!}
	{
		\begin{tabular}{|l|llll|llll|llll|llll|}
			\hline
			\multirow{2}{*}{} & \multicolumn{4}{c|}{Karate}   & \multicolumn{4}{c|}{Football} & \multicolumn{4}{c|}{Polbooks} & \multicolumn{4}{c|}{Amazon}          \\ \cline{2-17}
			& Purity & NMI   & ARI   & \#iters & Purity & NMI   & ARI   & \#iters & Purity & NMI   & ARI   & \#iters & Purity & NMI   & ARI   & \#iters        \\ \hline
			Attractor         & 1.000  & 0.924 & 0.939 & 13   & 0.930  & 0.924 & 0.888 & 9    & 0.857  & 0.589 & 0.680 & 16   & 0.978  & 0.960 & 0.580 & 62          \\
			IAttractor        & 0.529  & 0.000 & 0.000 & 6    & 0.783  & 0.638 & 0.846 & 7    & 0.467  & 0.000 & 0.000 & 7    & 0.716  & 0.846 & 0.033 & 8           \\ \hline
			{[}0.5-10{]}      & 1.000  & 0.924 & 0.939 & 11   & 0.930  & 0.924 & 0.888 & 9    & 0.857  & 0.589 & 0.680 & 13   & 0.978  & 0.960 & 0.580 & 18 $\downarrow$70.8\% \\
			{[}0.7-10{]}      & 1.000  & 0.924 & 0.939 & 12   & 0.930  & 0.924 & 0.888 & 9    & 0.857  & 0.589 & 0.680 & 15   & 0.978  & 0.960 & 0.580 & 25 $\downarrow$67.0\% \\ \hline
		\end{tabular}
	}
	\vspace{2pt}
	\caption{Performance of Attractor and sliding window on graphs with labels.}
	\label{tbl:slidingwindows_attractor_labeled_dataset}
	\vspace{-12pt}
\end{table*}

\begin{table*}[t]
	\centering
	\resizebox{1.0\linewidth}{!}
	{
		\begin{tabular}{|l|llll|llll|llll|llll|}
			\hline
			& \multicolumn{4}{c|}{Friendship}               & \multicolumn{4}{c|}{Amazon}                  & \multicolumn{4}{c|}{Collaboration}            & \multicolumn{4}{c|}{Road}                     \\ \cline{2-17}
			& modul & ncut  & \#coms & \#iters              & modul & ncut  & \#coms & \#iters             & modul & ncut  & \#coms & \#iters              & modul & ncut  & \#coms & \#iters              \\ \hline
			Attractor & 0.421 & 0.607 & 8044   & 323                  & 0.741 & 0.398 & 23822  & 62                  & 0.337 & 0.159 & 785    & 43                   & 0.865 & 0.264 & 56967  & 37                   \\ \hline
			0.5-10    & 0.347 & 0.606 & 7939   & 19  $\downarrow$35\% & 0.741 & 0.398 & 23800  & 18 $\downarrow$71\% & 0.337 & 0.158 & 784    & 17 $\downarrow$10\%  & 0.865 & 0.264 & 56952  & 18 $\downarrow$16\%  \\ \hline
			0.5-15    & 0.424 & 0.606 & 8022   & 23 $\downarrow$28\%  & 0.741 & 0.398 & 23820  & 23 $\downarrow$68\% & 0.337 & 0.158 & 784    & 23 $\downarrow$3.5\% & 0.865 & 0.264 & 56966  & 22 $\downarrow$14\%  \\ \hline
			0.7-10    & 0.416 & 0.606 & 7991   & 28 $\downarrow$36\%  & 0.741 & 0.398 & 23802  & 25 $\downarrow$69\% & 0.337 & 0.158 & 784    & 21 $\downarrow$9.4\% & 0.865 & 0.264 & 56952  & 21 $\downarrow$16\%  \\ \hline
			0.7-15    & 0.424 & 0.607 & 8034   & 35 $\downarrow$23\%  & 0.741 & 0.398 & 23821  & 31 $\downarrow$67\% & 0.337 & 0.158 & 784    & 27 $\downarrow$3.3\% & 0.865 & 0.264 & 56966  & 25 $\downarrow$12\%  \\ \hline
		\end{tabular}
	}
	\vspace{2pt}
	\caption{Performance of Attractor and sliding window on graphs without labels. }
	\label{tbl:slidingwindows_attractor_non_labeled_dataset}
	\vspace{-20pt}
\end{table*}

\subsection{Extracting communities.} When each of all edge distances converges to either 0 or 1, we remove edges with distance equal to 1. Then, we find communities as connected components by running the breath first search on the master node.

\section{Experiments}
\label{sec:experiments}

In this section, we evaluate performance of our sliding window model on real-world datasets. Then, we evaluate performance of MRAttractor on both synthetic and real-world datasets. Our experiments mainly focus on evaluating the efficiency because Attractor \cite{kdd.dyna} outperformed other community detection algorithms such as Spectral clustering, Louvain and Infomap, achieving higher NMI and ARI \cite{manning2008introduction}.

\subsection{Performance of Applying Sliding Window to Attractor}

\begin{table}
	\centering
	\resizebox{1.0\linewidth}{!}
	{
		\begin{tabular}{|l|l|l|l|l|l|}
			\hline
			\multicolumn{1}{|c|}{Datasets} & \multicolumn{1}{c|}{$|$V$|$} & \multicolumn{1}{c|}{$|$E$|$} & $|$classes$|$ & \multicolumn{1}{c|}{AVD} & \multicolumn{1}{c|}{CC} \\ \hline
			1M                             & 20,000                   & 1,000,310                & 191     & 100.031                  & 0.699                   \\ \hline
			2M                             & 40,000                   & 1,994,815                & 387     & 99.741                   & 0.692                   \\ \hline
			4M                             & 80,000                   & 3,996,488                & 761     & 99.912                   & 0.707                   \\ \hline
			6M                             & 120,000                  & 5,994,313                & 1137    & 99.905                   & 0.700                   \\ \hline
			8M                             & 160,000                  & 7,977,113                & 1512    & 99.713                   & 0.697                   \\ \hline
		\end{tabular}
	}
	\vspace{2pt}
	\caption{Synthetic networks}
	\label{tbl:tbl-fixed-100}
	\vspace{-15pt}
\end{table}


To evaluate performance of applying our sliding window to Attractor and whether the sliding window reduce the community detection quality, we utilized all the datasets employed in \cite{kdd.dyna}, including datasets with or without ground truth (knowing the number of true communities or not). Since these datasets were introduced in \cite{kdd.dyna}, we omit their descriptions in this paper. Table~\ref{tbl:sliding_window_datasets} presents statistics of the datasets. For labeled datasets, we report well-known measures - Purity, NMI and ARI \cite{manning2008introduction}, and \# of iterations (\#iters) to make algorithms converged. For unlabeled datasets, we report normalized cut (Ncut) \cite{shi2000normalized}, modularity \cite{blondel2008fast}, \# of extracted communities (\#coms), \# of iterations (\#iters), and running time reduction. 

We compared Attractor with our sliding window (SAttractor) with the original Attractor and IAttractor \cite{meng2016improved}, an improved version of Attractor to make edge distances converged quickly. We set $\lambda=0.5$ for all methods, following \cite{kdd.dyna}. For IAttractor, we used exactly same parameter settings in \cite{meng2016improved} such as enhanced factor $\delta=1$ and $Con\_Co$=0.99 (making edges converged when the proportion of converged edges was greater than $Con\_Co$ threshold). For sliding window, we set $\tau=\{0.5, 0.7\}$ since it is more reliable to reflect the converging trend of edges, and set $s=\{10, 15\}$.




Table \ref{tbl:slidingwindows_attractor_labeled_dataset} presents the results of Attractor, IAttractor and sliding window model with $s=10$ and $\tau\in\{0.5,0.7\}$ in labeled datasets. We only report results of $s=10$ because these datasets require a few iterations to converge and sliding window only take effect after $s$ iterations. The experimental results showed that (1) SAttractor achieved the same quality of extracted communities with Attractor; (2) The number of iterations took in SAttractor was always equal to or smaller than the number of iterations took in Attractor; and (3) IAttractor failed producing the same quality of communities. We investigated the reason why IAttractor performed poorly in terms of the quality, and found that factor $\delta$ increased dynamic interactions too much, and made most edges converged to zero.

\begin{table}
	\centering
	\resizebox{1.0\linewidth}{!}
	{
		\begin{tabular}{|l|l|l|c|l|l|}
			\hline
			\multicolumn{1}{|c|}{Datasets} & \multicolumn{1}{c|}{$|$V$|$} & \multicolumn{1}{c|}{$|$E$|$} & $|$classes$|$                      & \multicolumn{1}{c|}{AVD} & CC    \\ \hline
			DBLP                           & 317,080                      & 1,049,866                    & \multicolumn{1}{l|}{top5000} & 6.622                    & 0.632 \\ \hline
			Texas Road                     & 1,379,917                    & 1,921,660                    & unknown                           & 2.785                    & 0.047 \\ \hline
			Youtube                        & 1,134,890                    & 2,987,624                    & \multicolumn{1}{l|}{top5000} & 5.265                    & 0.081 \\ \hline
			Flixster                       & 2,523,386                    & 7,918,801                    & unknown                           & 6.276                    & 0.083 \\ \hline
		\end{tabular}
	}
	\vspace{2pt}
	\caption{The large-scale real-life networks. }
	\label{tbl:real-life-graphs-large}
	\vspace{-15pt}
\end{table}

In small datasets such as Karate, Football and Polbooks, we might observe minor improvement by applying the sliding window in terms of running time. However, in Amazon dataset, we observed huge improvement by reducing up to 70.8\% running time (reducing from 885 to 258 seconds). 

Table~\ref{tbl:slidingwindows_attractor_non_labeled_dataset} shows experimental results in unlabeled datasets. Note that IAttractor failed producing the same quality of communities again, so we only report results of Attractor and SAttractor. When we keep track up only the last 10 statues of edges (i.e., $s=10$), the quality of extracted communities in Friendship dataset was worse than the other sliding window settings. In other datasets, the quality was once again consistent with Attractor. When increasing $s$ to 15 and 20, the quality of extracted communities in all datasets were consistent with those found by Attractor. Besides, \#iterations was significantly reduced. In particular, \#iterations in Friendship dataset decreased from 323 to only 23 when we set $s=15$ and $\tau=0.5$, achieving 28\% running time reduction. 

\begin{figure*}[t]
	\centering
	\subfigure[Running time on synthetic graphs]{
		\label{fig:all_synthetic_network}
		\includegraphics[trim=5 50 100 60,clip,width=1.6in, height=1.1in]{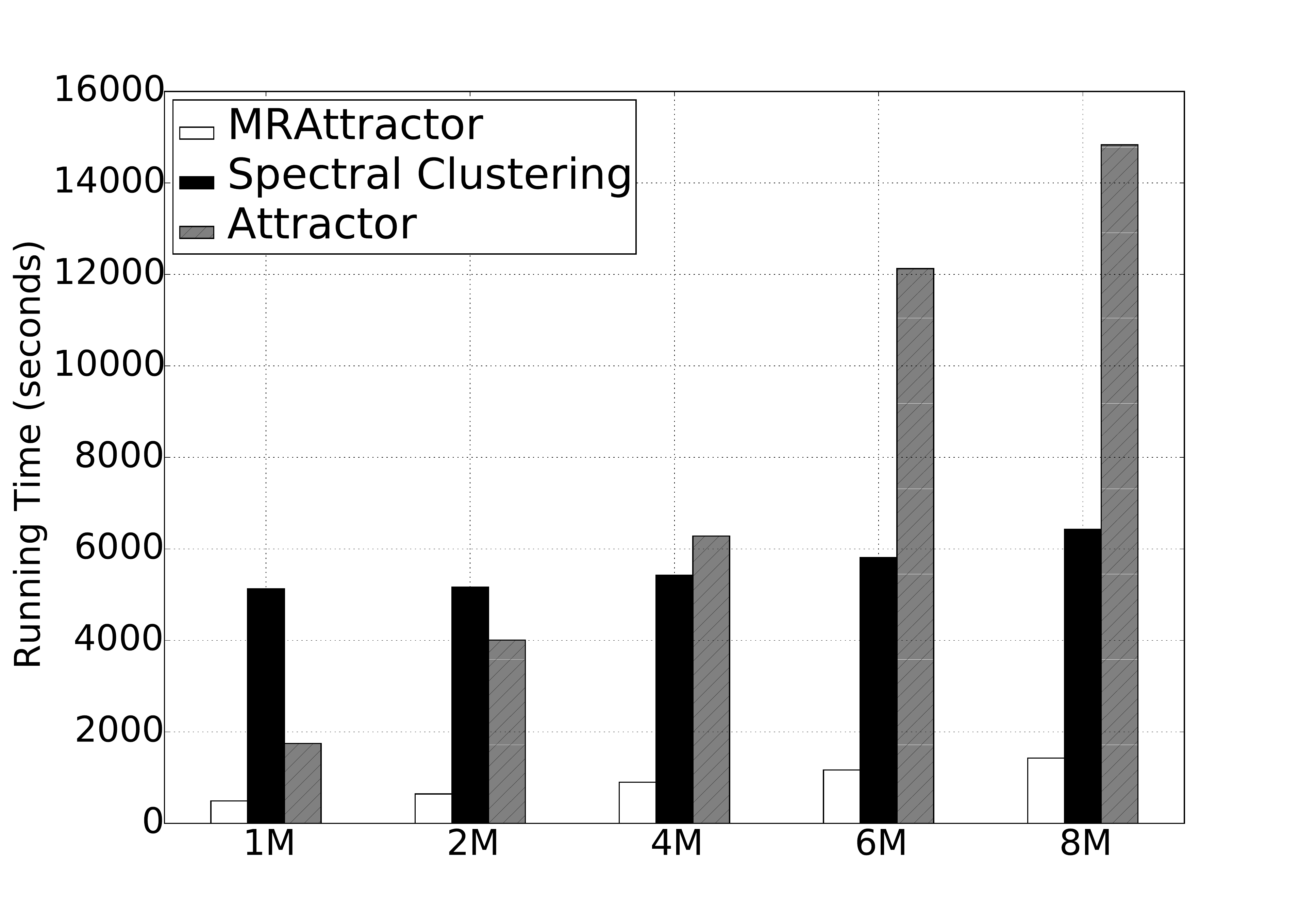}
	}
	\subfigure[Running time of the 1st iteration]{
		\label{fig:firstIterSynthetic}
		\includegraphics[trim=5 50 100 60,clip,width=1.6in, height=1.1in]{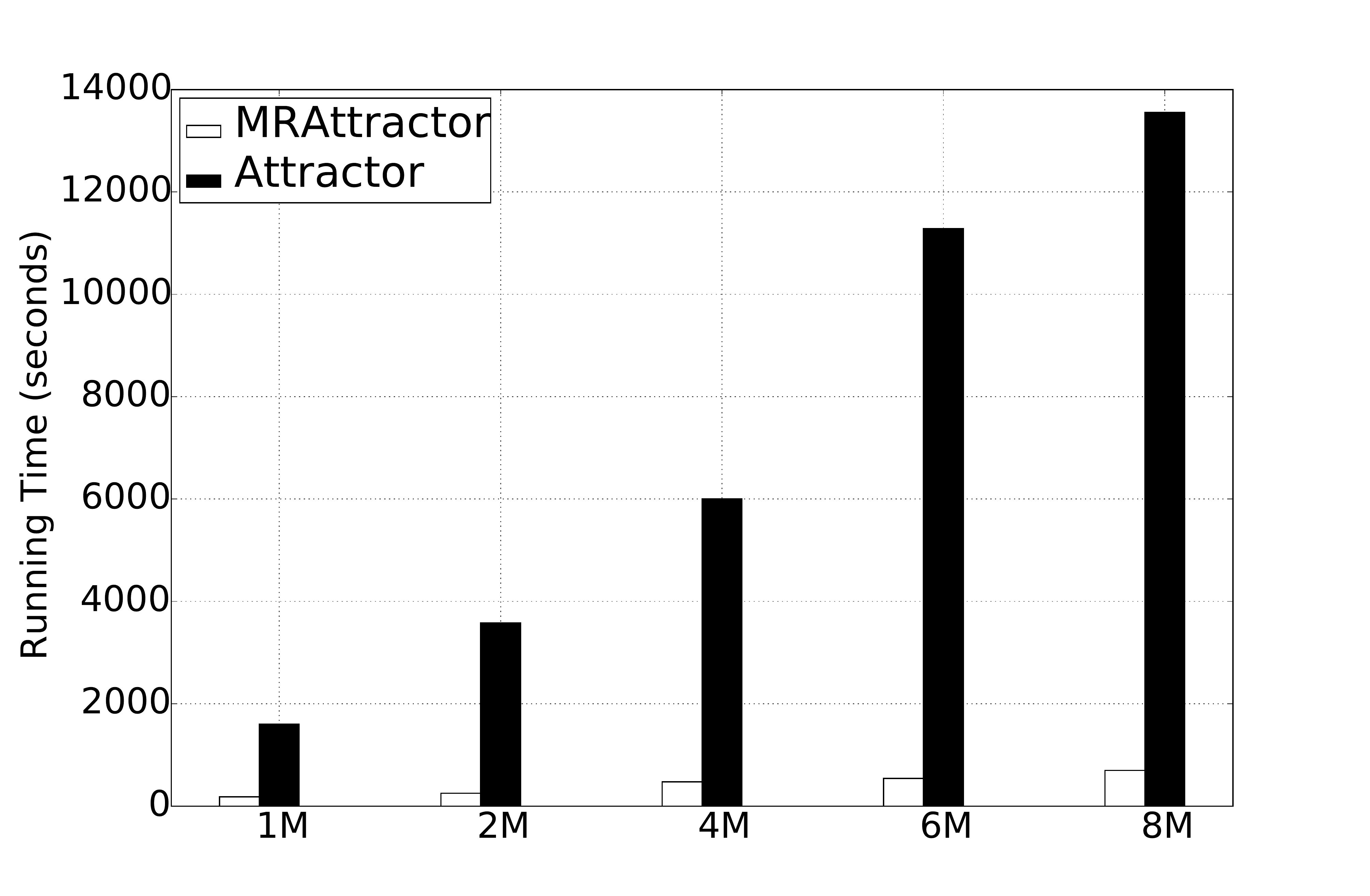}
	}
	\subfigure[Running time on real-life graphs]{
		\label{fig:all_real_life_dataset}
		\includegraphics[trim=5 55 100 70,clip,width=1.6in, height=1.1in] {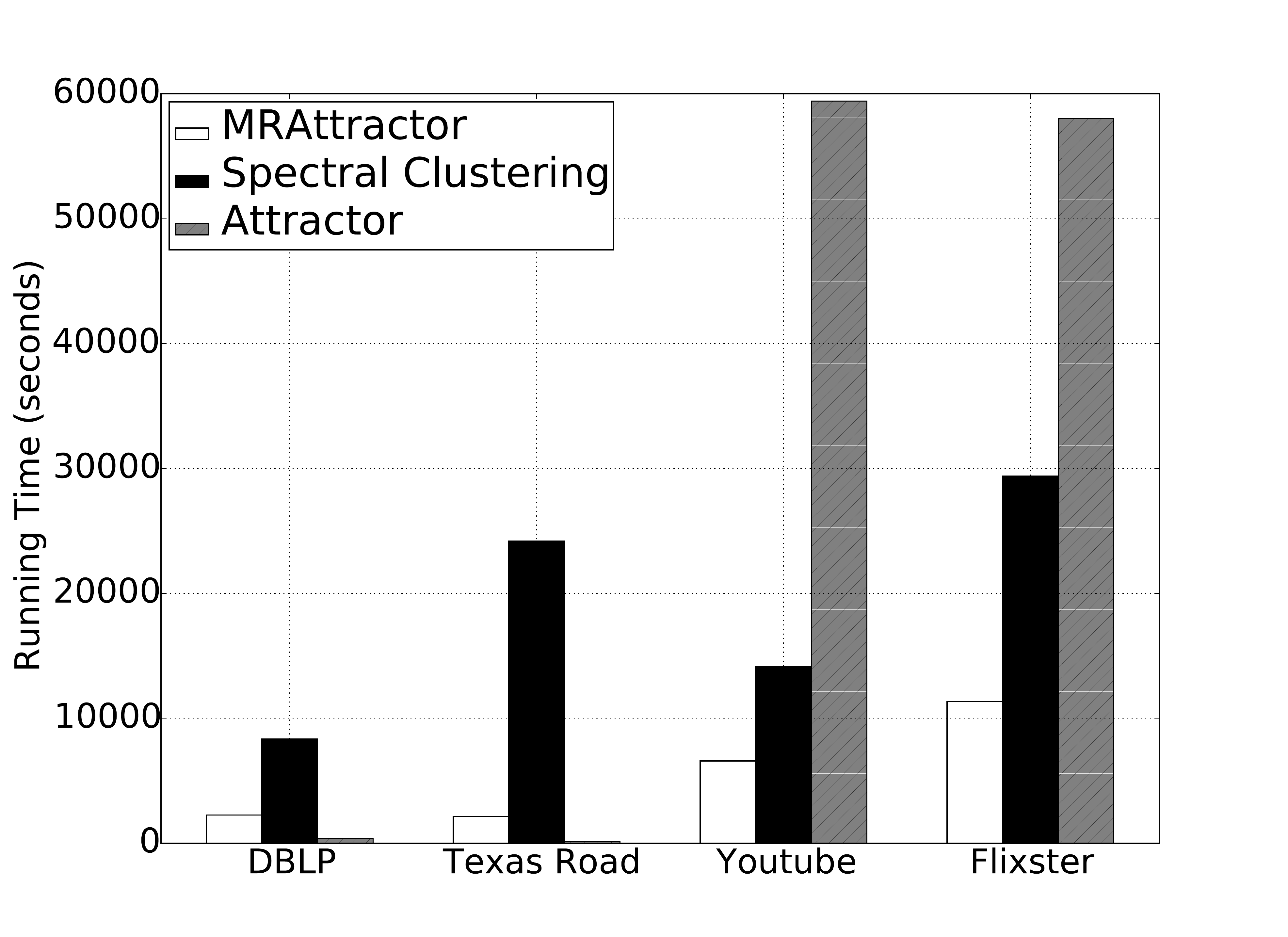}
	}
	\subfigure[Running time of the 1st iteration]{
		\label{fig:firstIterReal}
		\includegraphics[trim=5 55 100 70,clip,width=1.6in, height=1.1in]{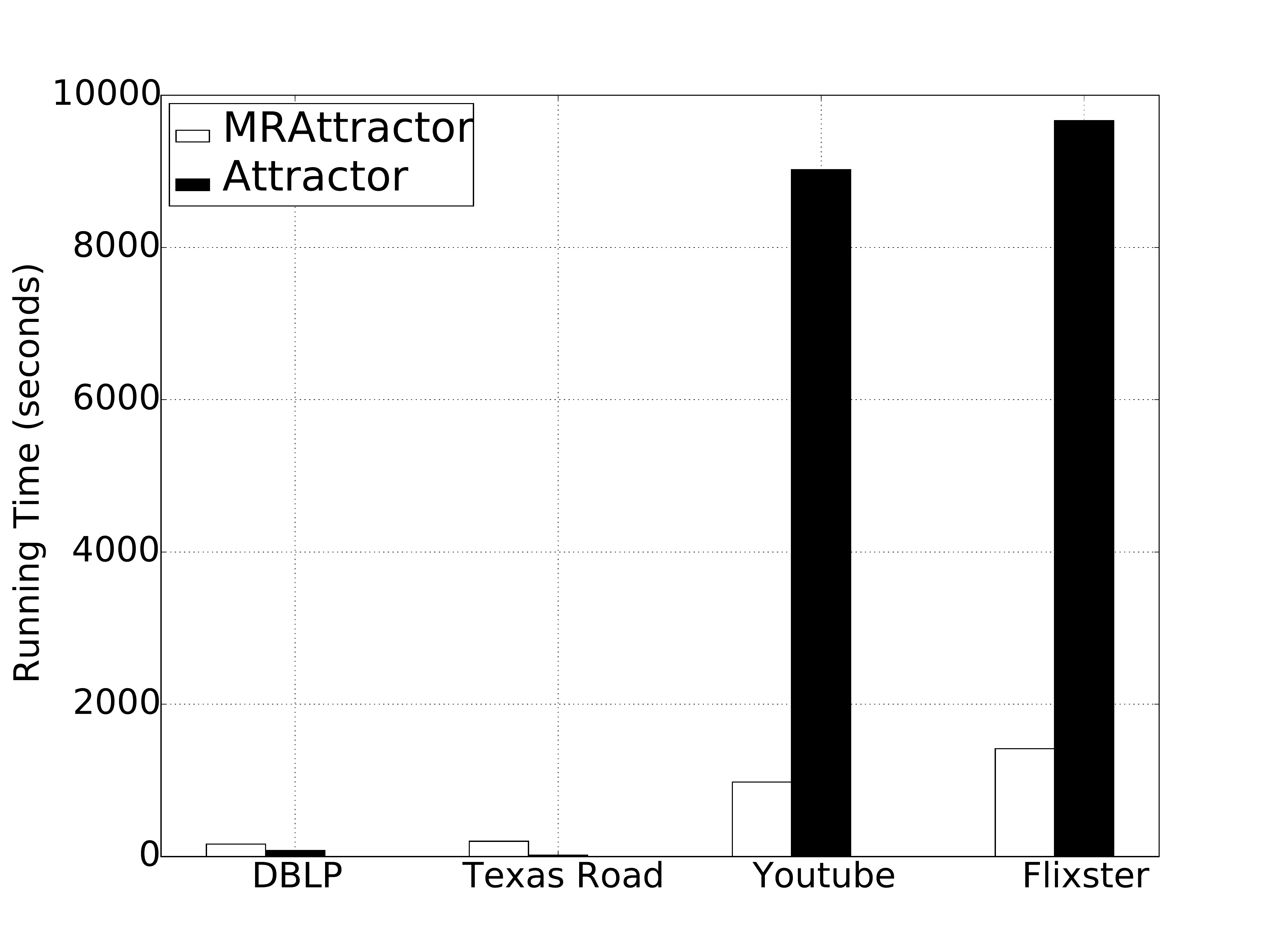}
	}
	\vspace{-5pt}
	\caption[]{Running time of MRAttractor with single-machine Attractor and Spectral Clustering algorithm on simulated networks and real-life graphs.}
	\label{fig:mrPerformance}
	\vspace{-15pt}
\end{figure*}

\subsection{Performance of \textsc{MRAttractor}}
So far, we observed that our sliding window model in a single machine environment reduced number of iterations and the running time, while preserved the quality of extracted communities. Next, we turn to examine MRAttractor's performance in large-scale graphs. As we described in the previous section, MRAttractor has several enhancements including the sliding window against Attractor. In the following experiments, we used both synthetic and real-life datasets/networks.

\subsubsection{Experimental settings}
We created five synthetic graphs presented in Table~\ref{tbl:tbl-fixed-100} by using Fortunato's benchmarking tool \cite{lancichinetti2008benchmark} for community detection. This tool was also employed in the prior work such as \cite{kdd.dyna, tsironis2013accurate, wickramaarachchi2014fast}. According to Twitter and Facebook statistics \cite{avg.twitter.follower,avg.facebook.friends}, users have 100s followers or friends on average. Therefore, when we generated synthetic datasets, we set average degree to 100. Table~\ref{tbl:real-life-graphs-large} shows statistics of large real-life graphs. 
\textbf{DBLP} is a graph of co-authorship where two authors are connected if they co-authored at least one paper. \textbf{Texas Road} is a road network where nodes are intersections and endpoints. Edges are roads connecting them in Texas. \textbf{Youtube} dataset is a snapshot of friendship graph on Youtube. If two users are friends, they will be connected. \textbf{Flixster} is the friendship network of movie site Flixster.com where two users are linked if they are friends.



We implemented \textsc{MRAttractor} on the top of Hadoop and ran experiments on a cluster which consisted of one master node and five slave nodes, each of which had 8 cores. We set the number of reducers to 30 and used sliding window with $s=15$ and $\tau=0.5$ since these settings achieved  the fastest running time while the quality of extracted communities was good in the previous experiment. The other parameters were $\lambda=0.5$, $\gamma=10000$, and the number of partitions $p=20$. The hashing function $P(\cdot)$ was modulo function.

We compared \textsc{MRAttractor} with two baselines:
\squishlist
\item \textbf{Attractor:} We implemented Attractor \cite{kdd.dyna} in Java and ran it on a computer with 64GB RAM and Intel Xeon 8 cores 2.10GHz. The cohesive parameter was $\lambda=0.5$.
\item \textbf{Spectral Clustering for MapReduce:} It was a MapReduce version of the spectral clustering proposed in \cite{tsironis2013accurate}. We set 20 iterations for finding eigenvectors and 30 iterations for K-means, following the same settings in \cite{tsironis2013accurate}.
\squishend

We ran \textsc{MRAttractor} and the baselines multiple times, and achieved consistent results.

\subsubsection{Experimental Results}

Figure~\ref{fig:mrPerformance} shows performance of MRAttractor and baselines in synthetic and real-life networks. In the synthetic networks (see Figure~\ref{fig:all_synthetic_network}), MRAttractor was significantly faster than Attractor and Spectral clustering. In particular, MRAttractor was 3.56$\sim$10.39 times faster than Attractor, and 4.50$\sim$10.45 times faster than Spectral clustering. Note that, we also generated synthetic graphs with average degree=200, and observed that MRAttractor was much faster than Attractor and Spectral clustering.

In real-life networks, MRAttractor performed worse than Attractor, when the datasets such as DBLP and Texas Road were small, and their average graph density was small as well. It makes sense because iterative MapReduce algorithms usually suffer overhead of network I/O and disk I/O \cite{jiang2010performance}. However, when MRAttractor dealt with larger datasets such as Youtube and Flixster, it performed significantly faster than the baselines. In particular, it was 5.12$\sim$9.02 times faster than Attractor (see Figure \ref{fig:all_real_life_dataset}). Attractor had 1,099 and 1,513 iterations to converge in Youtube and Flixster networks, respectively whereas MRAttractor only had 22 and 25 iterations in Youtube and Flixster, respectively.

In addition, Figures~\ref{fig:firstIterSynthetic} and~\ref{fig:firstIterReal} shows running time of the first iteration of both Attractor and MRAttractor in synthetic and real-life networks. The first iteration takes the longest time among all the iterations. Again, MRAttractor was much faster than Attractor across all the synthetic networks, and Youtube and Flixster. As we can observe in 1M and 2M synthetic networks, even though graph sizes in 1M and 2M synthetic networks were relatively small, as long as their average graph density was high, MRAttractor would be faster than Attractor (unlike DBLP and Texas Road networks which had the same edge size but much smaller average graph density).

\begin{figure}[h]
	\centering
	\subfigure[Speedup in the first 15 iterations]{
		\label{fig:speedup_15iter}
		\includegraphics[trim=5 50 90 60,clip,width=1.6in, height=1.1in]{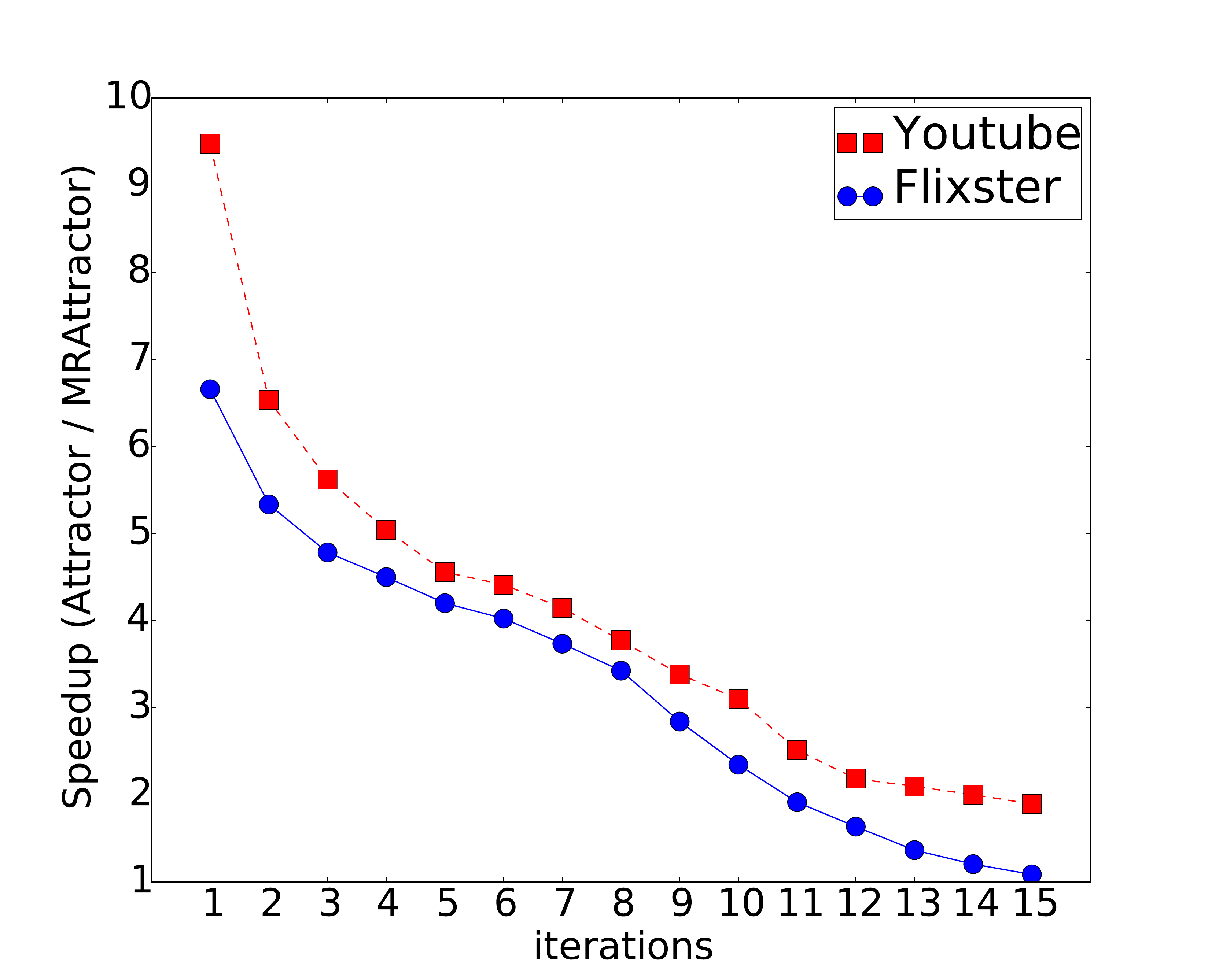}
	}
	\subfigure[Running time in each phase]{
		\label{fig:running_time_each_phase}
		\includegraphics[trim=5 45 90 55,clip,width=1.6in, height=1.1in]{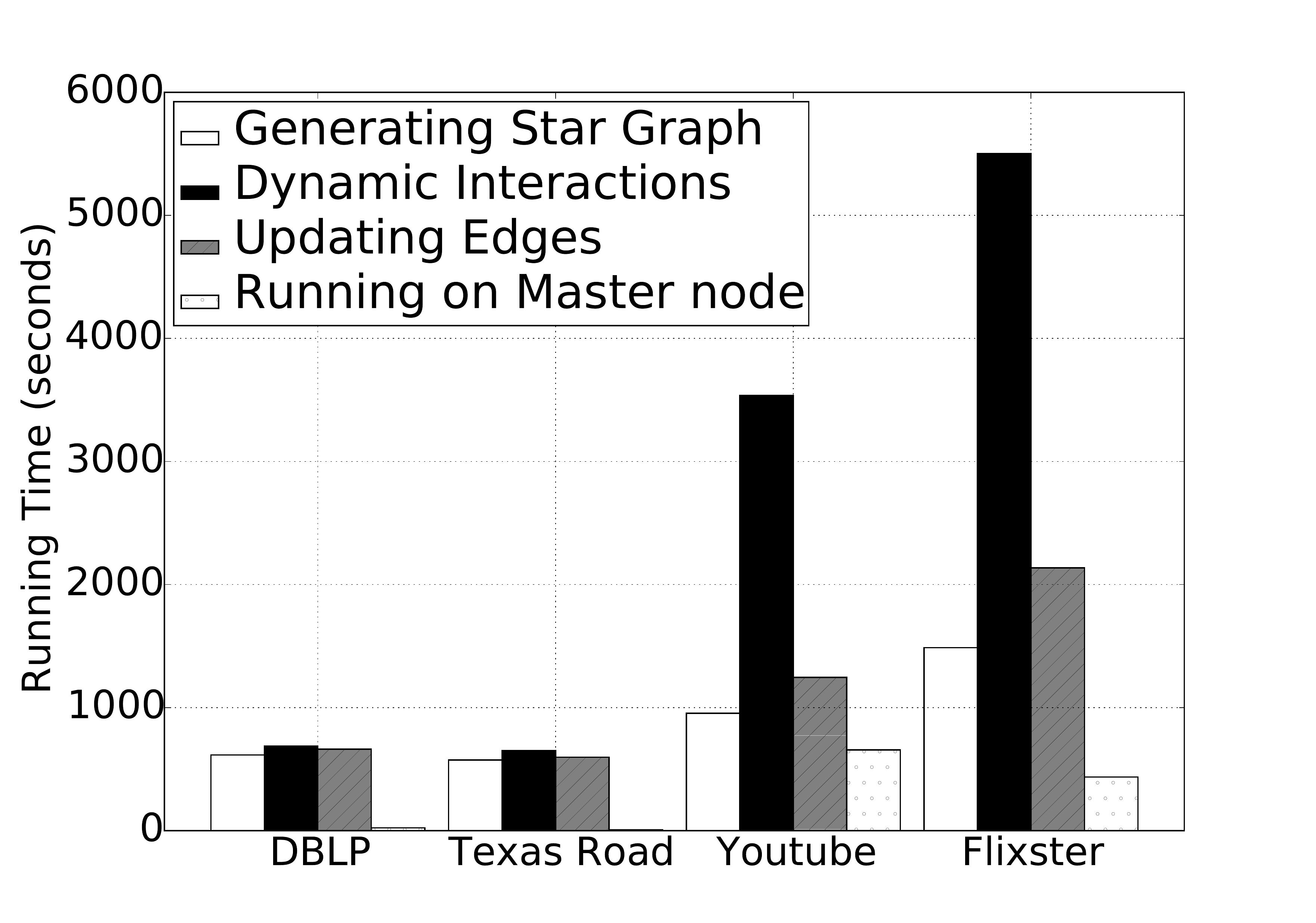}
	}
	\caption[]{Speedup of MRAttractor over Attractor (left) and running time of MRAttractor (right).}
\end{figure}

Figure~\ref{fig:speedup_15iter} shows how much speedup MRAttractor made over Attractor in Youtube and Flixster datasets in the first 15 iterations. 
MRAttractor was  9.5 and 6.8 times faster than Attractor in the first iteration. In later iterations, the speedup gradually decreased because many edges converged over time.

Next, we examine running time of each of the three MapReduce phases and remaining edge distance convergence on master node presented in Figure~\ref{fig:MRAttractor_major_components}. Figure~\ref{fig:running_time_each_phase} shows details of running time in each phase. In small datasets such as DBLP and Texas Road, running time of MapReduce phases were quite similar. However, in larger networks, running time of computing dynamic interactions became dominant due to high complexity of this step. We also observed that running time of remaining edge convergence step increased as the graph size increased even though we set small $\gamma=10,000$.

\subsubsection{Number of Graph Partitions and Reducers}
We are interested in measuring the sensitivity of the number of graph partitions $p$ in terms of running time. Figure~\ref{fig:sensitivity_no_partition} shows how running time in Flixster graph (the largest real-life dataset in this paper) was changed, when we varied $p\in[5;25]$. Given a small $p$, extracting communities was slow since each of subgraphs contained a large number of vertices and edges. As $p$ increased, the running time decreased significantly until $p$ was 14. When $p$ was greater than 14, it took longer running time because too many subgraphs were generated, leading to performance deterioration.

\begin{figure}
	\centering
	\subfigure[Effect of $p$ on running time.]{
		\label{fig:sensitivity_no_partition}
		\includegraphics[trim=5 10 90 55,clip,width=1.6in, height=1.2in]{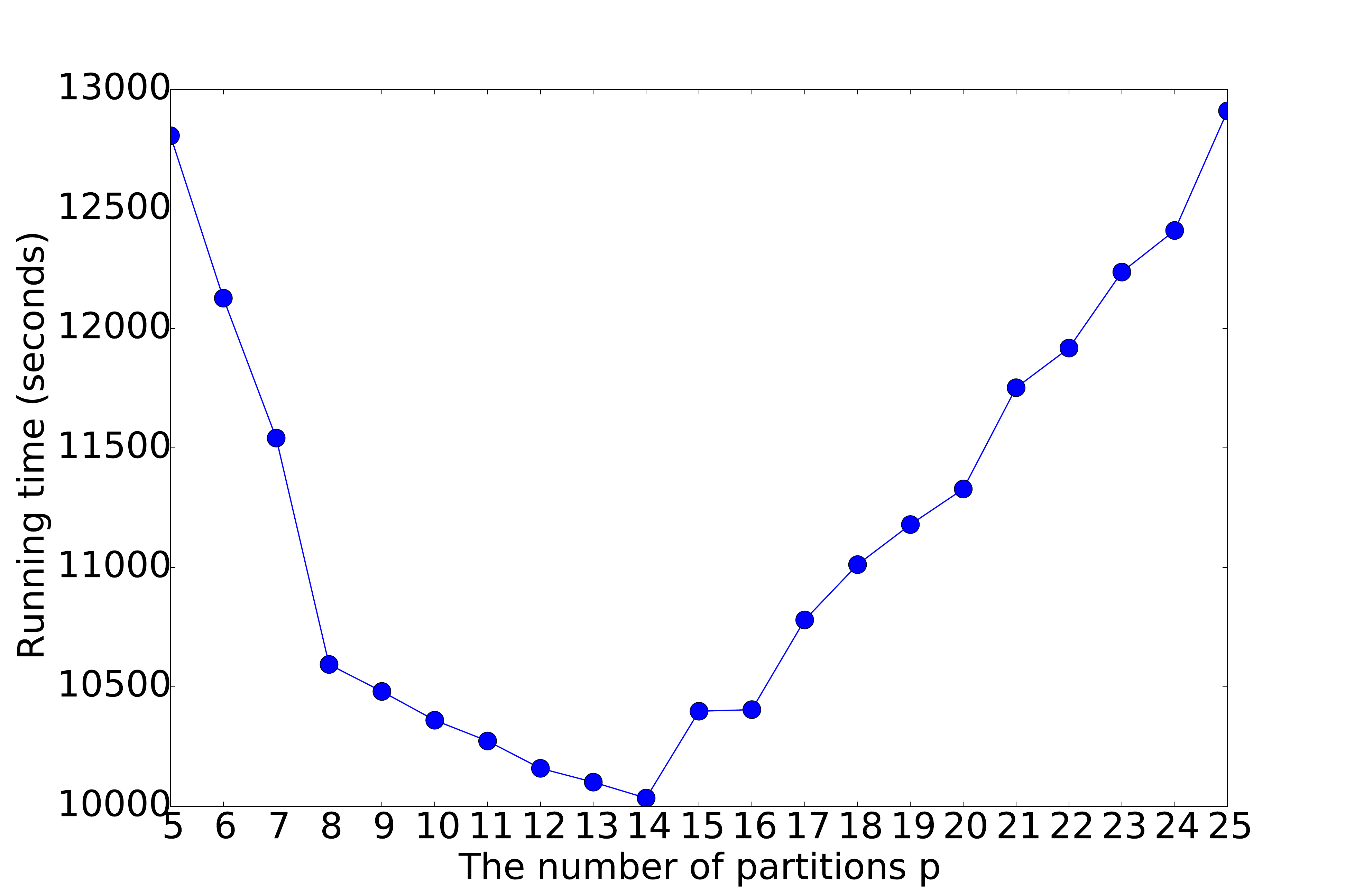}
	}
	\subfigure[Varying the number of reducers]{
		\label{fig:vary_reducers}
		\includegraphics[trim=5 10 90 55,clip,width=1.6in, height=1.2in]{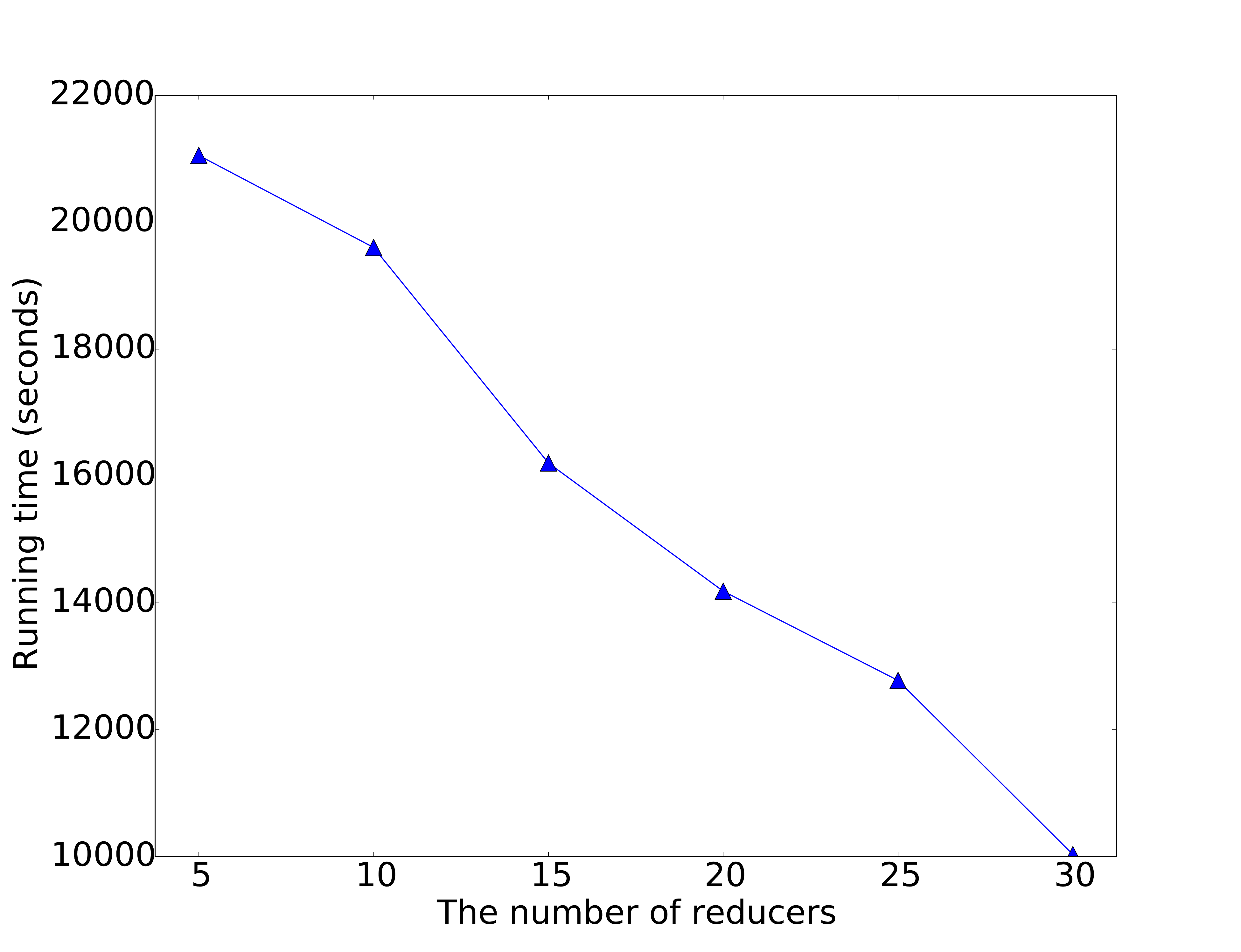}
	}
	\vspace{-5pt}
	\caption[]{Sensitivity of the number of graph partitions and reducers.}
	\vspace{-15pt}
\end{figure}
Another interesting question is how the number of reducers affects performance of MRAttractor. To answer the question, we set $p=14$,  and varied the number of reducers on Flixster dataset. Figure~\ref{fig:vary_reducers}, the running time decreased linearly as we increased the number of reducers in our Hadoop cluster.


Overall, MRAttractor not only produced the same quality of extracted communities like Attractor (we also ran other experiments to confirm MRAttractor correctly produced the same communities like Attractor), but also significantly faster than Attractor and a MapReduce version of Spectral Clustering in the large datasets or high graph density in small datasets.

\section{Conclusions and Future Work}
\label{sec:conclusion}
In this paper, we have presented how we have designed and implemented MRAttractor, an advanced version of Attractor for a MapReduce framework, Hadoop. Our proposed framework has handled large-scale graphs and significantly outperformed Attractor, IAttractor and Spectral Clustering, reducing running time and producing the same quality of extracted communities. In the future, we will implement MRAttractor for Spark and focus on handling large-scale time-evolving networks. 
Our graph partitioning algorithm has an advantage over existing graph processing frameworks when tackling exclusive interactions. In Pregel \cite{malewicz2010pregel}, for example, a vertex cannot directly communicate with an exclusive neighbor $x$ without knowing $x$'s identifier. GraphX \cite{gonzalez2014graphx}, even does not support communication of unconnected nodes. Thus, when designing algorithms that require exclusive interactions, our method can be customized to meet this demand.
\section*{Acknowledgment}
This work was supported in part by NSF grants CNS-1553035 and CNS-1755536, and Google Faculty Research Award. Any opinions, findings and conclusions or recommendations expressed in this material are the author(s) and do not necessarily reflect those of the sponsors.

\bibliographystyle{IEEEtran}
\bibliography{www}

\end{document}